\documentclass[12pt,titlepage]{utarticle}
\usepackage{amsmath}
\usepackage{mathrsfs}
\usepackage{amsfonts}
\usepackage{amssymb}
\usepackage{amsthm}
\usepackage{mathtools}
\usepackage{graphicx}
\usepackage{color}
\usepackage{cite}
\usepackage{enumitem}
\usepackage{ucs}
\usepackage{xparse}
\usepackage{tikz}
\usepackage{hyperref}
\usetikzlibrary{decorations.pathreplacing,calligraphy}
\usepackage{comment}
\usepackage{float}

\usepackage{kpfonts}
\usepackage[scaled=0.90]{helvet} 
\usepackage{courier} 
\normalfont
\usepackage[T1]{fontenc}

\definecolor{aqua}{rgb}{0, 1.0, 1.0}
\definecolor{fuschia}{rgb}{1.0, 0, 1.0}
\definecolor{gray}{rgb}{0.502, 0.502, 0.502}
\definecolor{lime}{rgb}{0, 1.0, 0}
\definecolor{maroon}{rgb}{0.502, 0, 0}
\definecolor{navy}{rgb}{0, 0, 0.502}
\definecolor{olive}{rgb}{0.502, 0.502, 0}
\definecolor{purple}{rgb}{0.502, 0, 0.502}
\definecolor{silver}{rgb}{0.753, 0.753, 0.753}
\definecolor{teal}{rgb}{0, 0.502, 0.502}

\makeatletter
\newdimen\itex@wd%
\newdimen\itex@dp%
\newdimen\itex@thd%
\def\itexspace#1#2#3{\itex@wd=#3em%
\itex@wd=0.1\itex@wd%
\itex@dp=#2ex%
\itex@dp=0.1\itex@dp%
\itex@thd=#1ex%
\itex@thd=0.1\itex@thd%
\advance\itex@thd\the\itex@dp%
\makebox[\the\itex@wd]{\rule[-\the\itex@dp]{0cm}{\the\itex@thd}}}
\makeatother

\makeatletter
\newif\if@sup
\newtoks\@sups
\def\append@sup#1{\edef\act{\noexpand\@sups={\the\@sups #1}}\act}%
\def\reset@sup{\@supfalse\@sups={}}%
\def\mk@scripts#1#2{\if #2/ \if@sup ^{\the\@sups}\fi \else%
  \ifx #1_ \if@sup ^{\the\@sups}\reset@sup \fi {}_{#2}%
  \else \append@sup#2 \@suptrue \fi%
  \expandafter\mk@scripts\fi}
\def\tensor#1#2{\reset@sup#1\mk@scripts#2_/}
\def\multiscripts#1#2#3{\reset@sup{}\mk@scripts#1_/#2%
  \reset@sup\mk@scripts#3_/}
\makeatother

\makeatletter
\newbox\slashbox \setbox\slashbox=\hbox{$/$}
\def\itex@pslash#1{\setbox\@tempboxa=\hbox{$#1$}
  \@tempdima=0.5\wd\slashbox \advance\@tempdima 0.5\wd\@tempboxa
  \copy\slashbox \kern-\@tempdima \box\@tempboxa}
\def\slash{\protect\itex@pslash}
\makeatother

\def\clap#1{\hbox to 0pt{\hss#1\hss}}

\let\oldroot\root
\def\root#1#2{\oldroot #1 \of{#2}}
\renewcommand{\sqrt}[2][]{\oldroot #1 \of{#2}}

\DeclareSymbolFont{symbolsC}{U}{txsyc}{m}{n}
\SetSymbolFont{symbolsC}{bold}{U}{txsyc}{bx}{n}
\DeclareFontSubstitution{U}{txsyc}{m}{n}

\DeclareSymbolFont{stmry}{U}{stmry}{m}{n}
\SetSymbolFont{stmry}{bold}{U}{stmry}{b}{n}

\DeclareFontFamily{OMX}{MnSymbolE}{}
\DeclareSymbolFont{mnomx}{OMX}{MnSymbolE}{m}{n}
\SetSymbolFont{mnomx}{bold}{OMX}{MnSymbolE}{b}{n}
\DeclareFontShape{OMX}{MnSymbolE}{m}{n}{
    <-6>  MnSymbolE5
   <6-7>  MnSymbolE6
   <7-8>  MnSymbolE7
   <8-9>  MnSymbolE8
   <9-10> MnSymbolE9
  <10-12> MnSymbolE10
  <12->   MnSymbolE12}{}

\makeatletter
\def\re@DeclareMathSymbol#1#2#3#4{%
    \let#1=\undefined
    \DeclareMathSymbol{#1}{#2}{#3}{#4}}
\re@DeclareMathSymbol{\neArrow}{\mathrel}{symbolsC}{116}
\re@DeclareMathSymbol{\neArr}{\mathrel}{symbolsC}{116}
\re@DeclareMathSymbol{\seArrow}{\mathrel}{symbolsC}{117}
\re@DeclareMathSymbol{\seArr}{\mathrel}{symbolsC}{117}
\re@DeclareMathSymbol{\nwArrow}{\mathrel}{symbolsC}{118}
\re@DeclareMathSymbol{\nwArr}{\mathrel}{symbolsC}{118}
\re@DeclareMathSymbol{\swArrow}{\mathrel}{symbolsC}{119}
\re@DeclareMathSymbol{\swArr}{\mathrel}{symbolsC}{119}
\re@DeclareMathSymbol{\nequiv}{\mathrel}{symbolsC}{46}
\re@DeclareMathSymbol{\Perp}{\mathrel}{symbolsC}{121}
\re@DeclareMathSymbol{\Vbar}{\mathrel}{symbolsC}{121}
\re@DeclareMathSymbol{\sslash}{\mathrel}{stmry}{12}
\re@DeclareMathSymbol{\boxslash}{\mathrel}{stmry}{27}
\re@DeclareMathSymbol{\boxbslash}{\mathrel}{stmry}{28}
\re@DeclareMathSymbol{\boxast}{\mathrel}{stmry}{24}
\re@DeclareMathSymbol{\boxcircle}{\mathrel}{stmry}{29}
\re@DeclareMathSymbol{\boxbox}{\mathrel}{stmry}{30}
\re@DeclareMathSymbol{\obslash}{\mathrel}{stmry}{20}
\re@DeclareMathSymbol{\obar}{\mathrel}{stmry}{58}
\re@DeclareMathSymbol{\olessthan}{\mathrel}{stmry}{60}
\re@DeclareMathSymbol{\ogreaterthan}{\mathrel}{stmry}{61}
\re@DeclareMathSymbol{\bigsqcap}{\mathop}{stmry}{"64}
\re@DeclareMathSymbol{\biginterleave}{\mathop}{stmry}{"6}
\re@DeclareMathSymbol{\invamp}{\mathrel}{symbolsC}{77}
\re@DeclareMathSymbol{\parr}{\mathrel}{symbolsC}{77}
\makeatother

\makeatletter
\def\Decl@Mn@Delim#1#2#3#4{%
  \if\relax\noexpand#1%
    \let#1\undefined
  \fi
  \DeclareMathDelimiter{#1}{#2}{#3}{#4}{#3}{#4}}
\def\Decl@Mn@Open#1#2#3{\Decl@Mn@Delim{#1}{\mathopen}{#2}{#3}}
\def\Decl@Mn@Close#1#2#3{\Decl@Mn@Delim{#1}{\mathclose}{#2}{#3}}
\Decl@Mn@Open{\llangle}{mnomx}{'164}
\Decl@Mn@Close{\rrangle}{mnomx}{'171}
\Decl@Mn@Open{\lmoustache}{mnomx}{'245}
\Decl@Mn@Close{\rmoustache}{mnomx}{'244}
\Decl@Mn@Open{\llbracket}{stmry}{'112}
\Decl@Mn@Close{\rrbracket}{stmry}{'113}
\makeatother

\makeatletter
\DeclareRobustCommand\widecheck[1]{{\mathpalette\@widecheck{#1}}}
\def\@widecheck#1#2{%
    \setbox\z@\hbox{\m@th$#1#2$}%
    \setbox\tw@\hbox{\m@th$#1%
       \widehat{%
          \vrule\@width\z@\@height\ht\z@
          \vrule\@height\z@\@width\wd\z@}$}%
    \dp\tw@-\ht\z@
    \@tempdima\ht\z@ \advance\@tempdima2\ht\tw@ \divide\@tempdima\thr@@
    \setbox\tw@\hbox{%
       \raise\@tempdima\hbox{\scalebox{1}[-1]{\lower\@tempdima\box
\tw@}}}%
    {\ooalign{\box\tw@ \cr \box\z@}}}
\makeatother

\makeatletter
\NewDocumentCommand\mathraisebox{moom}{%
\IfNoValueTF{#2}{\def\@temp##1##2{\raisebox{#1}{$\m@th##1##2$}}}{%
\IfNoValueTF{#3}{\def\@temp##1##2{\raisebox{#1}[#2]{$\m@th##1##2$}}%
}{\def\@temp##1##2{\raisebox{#1}[#2][#3]{$\m@th##1##2$}}}}%
\mathpalette\@temp{#4}}
\makeatletter

\makeatletter
\def\udots{\mathinner{\mkern2mu\raise\p@\hbox{.}
\mkern2mu\raise4\p@\hbox{.}\mkern1mu
\raise7\p@\vbox{\kern7\p@\hbox{.}}\mkern1mu}}
\makeatother

\newcommand{\gt}{>}
\newcommand{\lt}{<}

\renewcommand{\scriptsize}{\scriptstyle}

\theoremstyle{plain}
\newtheorem{theorem}{Theorem}
\newtheorem{lemma}{Lemma}
\newtheorem{prop}{Proposition}

\theoremstyle{definition}
\newtheorem{defn}{Definition}
\newtheorem{conjecture}{Conjecture}

\theoremstyle{remark}
\newtheorem{remark}{Remark}

\newcommand{\charpol}{\operatorname{charpol}}

\begin{document}

\preprint{UTWI-37-2023\\}

\title{The Hitchin Image in Type-D}

\author{
 Aswin Balasubramanian%
 \address{
     International Centre for Theoretical Sciences\\
     Survey No. 151, Shivakote\\
     Hesaraghatta Hobli\\
     Bengaluru, 560 089 India\\
     \email{aswinb.phys@gmail.com}
},
 Jacques Distler%
 \address{
     Theory Group\\
     Department of Physics,\\
     University of Texas at Austin,\\
     Austin, TX 78712, USA \\
      \email{distler@golem.ph.utexas.edu}\\
      \email{cjp3247@utexas.edu}
 },
Ron Donagi%
\address{
    Department of Mathematics\\
    University of Pennsylvania\\
    Philadelphia, PA 19104-6395, USA\\
     \email{donagi@math.upenn.edu}\\
},
Carlos Perez-Pardavila${}^\mathrm{b}$ 
}

\Abstract{Motivated by their appearance as Coulomb branch geometries of Class S theories, we study the image of the local Hitchin map in tame Hitchin systems of type-D with residue in a special nilpotent orbit $\mathcal{O}_H$. We describe two important features which distinguish it from the type A case studied in \cite{Balasubramanian:2020fwc}. The first feature, which we term \textit{even type constraints}, arises iff the partition label $[\mathcal{O}_H]$ has even parts. Our Hitchin image is non-singular and thus different from the one studied by Baraglia and Kamgarpour.  We argue that our Hitchin image always globalizes to being the Hitchin base of an integrable system. The second feature, which we term  \textit{odd type constraints}, is related to a particular finite group $\overline{A}_b(\mathcal{O}_H)$ being non-trivial. This finite group parametrizes the choices for the local Hitchin base. Additionally, we also show that the finite group $\overline{A}_b(\mathcal{O}_H)$ encodes the size of the dual special piece. } 

 \date{October 9, 2023.}
\maketitle 
\tableofcontents
 

\section{Introduction}

One place where complex integrable systems enter in Physics is in describing the geometry of the Coulomb branch of a 4D $\mathcal{N}=2$ supersymmetric quantum field theory. The Coulomb branch is a special-K\"ahler manifold, whose geometry is governed by a complex integrable system \cite{Freed:1997dp,Donagi:1995cf}. In the case of class-S theories, that complex integrable system is a Hitchin system of type ADE \cite{Gaiotto:2009we,Gaiotto:2009hg}. The base curve, $C$, is the (punctured) Riemann surface on which the 6D (2,0) theory of type ADE is compactified.

The \emph{superconformal} theories of class-S (where all masses are tuned to zero) are governed by particular symplectic leaves of the Hitchin integrable system. In turn, if $C$ has genus-$g$ and $n$ punctures, the corresponding superconformal theories come in a \emph{family} of complex dimension $3g-3+n$, parameterized by the complex structure of the punctured curve $C$. We are thus naturally led to consider \emph{families} of Hitchin systems, fibered over the moduli space of complex structures, $\mathcal{M}_{g,n}$. From the Physics, it is natural to extend this to the boundary, where $C$ degenerates.  In \cite{Balasubramanian:2020fwc}, we studied this for type-A, and showed how the Hitchin bases fit together to form a holomorphic vector bundle over $\overline{\mathcal{M}}_{g,n}$. For $\overline{\mathcal{M}}_{0,4}\simeq \mathbb{CP}^1$, this vector bundle can be calculated very explicitly, and its splitting type\footnote{Every holomorphic vector bundle, $E\to\mathbb{CP}^1$ splits as a direct sum of line bundles, and the degrees of those line bundles (the ``splitting type'') completely characterize $E$.} constitutes a new global property of the family of SCFTs.

Extending that construction to types $D$ and $E$ is the subject of the present series of papers. In this first installment, we will study the \emph{local} behaviour of the Higgs field near one of the punctures, where it is assumed to have a simple pole with nilpotent residue. \textit{Therefore, and unless otherwise instructed, we'll always refer to the local Hitchin base.}

In type-A, the \emph{local Hitchin base} is parameterized by the Laurent tails of the coefficients of powers of $\lambda$ in $\charpol(\phi)=\det(\phi(t)-\lambda I)$. The new ingredient in type-D is that the \emph{leading} Laurent coefficients are not all independent; rather, they obey certain polynomial \emph{constraints}. Those constraints were first computed in \cite{Chacaltana:2011ze}. Later, Baraglia and Kamgarpour \cite{MR3815160} found a general expression for the constraints when the residue of $\phi$ lies in a \emph{Richardson} nilpotent orbit\footnote{A Richardson nilpotent orbit is one that is induced from the zero orbit inside some Levi subalgebra \cite{CollingwoodMcGovern}.}. Part of the purpose of the present paper is to extend that to the (physically-relevant) case of an arbitrary \emph{special} nilpotent orbit (see Def \ref{def_special_orbit}). Outside of Cartan type A, Richardson nilpotent orbits form a proper subset of the set of special nilpotent orbits. 

Moreover, there are a few crucial differences between the constructions of Chacaltana-Distler (CD) and Baraglia-Kamgarpour (BK). They construct different Hitchin bases. CD's base is birational to (a finite cover of) BK's base. This will have a dramatic effect when we go to construct global examples, as we briefly do in Section \ref{globalconsiderations}. 

In Section \ref{sec:local_hitchin_base}, we outline the local Hitchin setup. We review Spaltenstein's results on the factorization of $\charpol(\phi)$ in Section \ref{sec:spaltenstein}. The factorization that Spaltenstein found is at the root of what we call the ``even type" constraints of Baraglia-Kamgarpour \cite{MR3815160}  and Chacaltana-Distler \cite{Chacaltana:2011ze} (who called them ``$c$-type constraints''). Those two sets of authors take different approaches to implementing these constraints and the difference is summarized in Theorem \ref{th:BKCD} of Section \ref{sec:relBKCD}.

In addition to the even type constraints, Spaltenstein's factorization results also lead to a set of what we call ``odd type'' constraints (Chacaltana-Distler called these ``$a$-type constraints''). We explain this in Section \ref{sec:oddtype}, where we relate the existence of odd type constraints to the presence of a nontrivial dual special piece (see Def \ref{specialpiece-defn} ), dual to the Hitchin nilpotent orbit. Each odd type constraint is associated with passing to a double-cover (setting $c_{2k}=a_k^2$ in the notation below) of the naive Hitchin base. More precisely, we define a subgroup of Lusztig's canonical quotient, $\overline{A}_b\subset \overline{A}(\mathcal{O}_H)$ such that the nilpotent orbits in the dual special piece are in 1-1 correspondence with subgroups $C\subset \overline{A}_b$. If (as is the case in type-D) $\overline{A}_b =(\mathbb{Z}_2)^s$, then there are $2^s$ \emph{choices} for the Hitchin base, which has a factor $\mathbb{C}^s/C$. These choices are labeled by the choice of a subgroup $C\subset \overline{A}_b$ or equivalently by a choice of nilpotent orbit in the dual special piece. 

Both part (iii) of Theorem \ref{th:BKCD} and the existence of the odd-type constraints discussed in Section \ref{sec:oddtype} rely on a non-obvious feature of the coefficients that appear in Spaltenstein's factorization of $\charpol(\phi)$. A-priori, the coefficients are algebraic (in particular, not necessarily single-valued)  functions on the Lie algebra. Certain of those coefficients, however,  turn out to be polynomial functions on the Lie algebra. The result is actually independent of the details of Spaltenstein's factorization and can be summarized in Conjecture \ref{Simpleconjecture}. As discussed in Section \ref{sec:main_conjecture}, we have a lot of evidence for, but no proof of, the conjecture.

In Section \ref{sec:poisson}, we explain how the present work is compatible with the existence of mass deformation (or Poisson deformations) of the integrable system associated to the superconformal field theory. As explained in \cite{Balasubramanian:2018pbp}, these deformations, in their local incarnation, correspond to \textit{sheets} in a Lie algebra. In Section \ref{sec:poisson}, we outline a conjecture that relates the local Hitchin image and Losev's theory of birational sheets \cite{MR4359565}, which is refinement of the usual theory of sheets. 

Finally, in Section \ref{globalconsiderations}, we give some global examples which illustrate the difference between the Hitchin bases constructed \`a la Chacaltana-Distler versus Baraglia-Kamgarpour and the different Hitchin bases that result from different choices of odd type constraints.

\section{The Local Hitchin Base in Type-D}
\label{sec:local_hitchin_base}

We study the local Hitchin base for tame Higgs bundles of type $D_n$ with special nilpotent residues for the Higgs field. In this paper, we work in the framework of meromorphic Higgs bundles \cite{markman1994spectral,bottacin1995symplectic}. 

Our primary motivation is to extend the work in \cite{Balasubramanian:2020fwc} where we studied Hitchin base(s) of Cartan type A as a \emph{family} over the moduli space of curves. In order to pursue a similar study in type-D, it turns out we need to understand some new phenomena that occur already in the study of the local Hitchin base. One can, in fact, study the Hitchin base at the following three levels

\begin{enumerate}%
\item The purely local Hitchin base/Hitchin image (for a punctured disc).
\item The global Hitchin base (for a global curve $C_{g,n}$ ).
\item The family of Hitchin bases over the moduli of curves $\overline{\mathcal{M}}_{g,n}$ (``very global'').

\end{enumerate}

In what follows, we confine ourselves to questions that can be addressed at the local level. The stories are, however, tightly linked and at a certain point, we will actually have to address the question of which features in the local story globalize. We will separately consider the new questions arising at the global and very global levels.

At the purely local level, when compared to type A, we find that describing the image of the Hitchin map in type-D is more subtle. Let $\mathcal{B}^{\text{naive}}$ be the naive Hitchin base computed as though we were in type A (by this, we mean that we follow the Young tableaux based algorithm in \cite{Balasubramanian:2020fwc} and just ignore the base contributions in odd degrees except for the Pfaffian in the case of $D_N$, $N$ odd). We say the Hitchin base \emph{has constraints} iff the actual local Hitchin base $\mathcal{B}$ is not the same as $\mathcal{B}^{\text{naive}}$.

Recall that nilpotent orbits in $D_n$ are classified by D-type partitions of $2n$\cite{CollingwoodMcGovern}. D-type partitions are those partitions in which even parts occur with even multiplicity. We also follow the usual convention that the parts $p_j$ of a partition $[p_j]$ are always written in non-increasing order. 

The constraints we encounter come in two types
\begin{enumerate}
\item \textbf{Even-type constraints} : These occur \emph{if and only if} the Hitchin partition has even parts.
\item \textbf{Odd-type constraints} : These occur \emph{if and only if} the $2j^{\text{th}}$ and $(2j+1)^{\text{st}}$ parts of the Hitchin partition are both odd ($p_{2j}=2r+1$, $p_{2j+1}=2s+1$) with $p_{2j}> p_{2j+1}$.
\end{enumerate}


When any of these constraints occur, the description of the Hitchin base requires more work. However, in \emph{all} these cases, we find that the actual Hitchin base is affine and smooth. This is in contrast to the results of \cite{MR3815160} where they find instances where the local Hitchin base in singular. Locally, the difference arises from how one chooses the generators of the (local) ring of invariant polynomials. Our choice is natural from a physical perspective and was the one that was made in the examples studied in \cite{Chacaltana:2011ze}. 

We restrict to special nilpotent residues for the Higgs field for a few reasons. First, these are the orbits that have appeared so far in the description of of the Coulomb branches of superconformal field theories of class-S. Second, we will find that aspects of the story about even type constraints are easier to understand for special nilpotent residues. Even type constraints can be thought of as an extension of what was called \emph{fingerprints of surface operators} in \cite{Gukov:2006jk}. 


\begin{remark}
    In the GIT framework for constructing the moduli of stable Higgs bundles, an important notion is that of \textit{strongly parabolic Higgs bundles}. They correspond to the cases where the residue of the Higgs field is nilpotent and compatible with the parabolic structure at the puncture. From our point of view, the formalism of strongly parabolic Higgs bundles is somewhat limiting as it leads only to residues that live in Richardson nilpotent orbits (as in \cite{MR3815160}). For our purposes, we needed to allow all special nilpotent orbits as possible residues. For any simple Lie algebra not of type-$A$, the set of Richardson nilpotent orbits form a proper subset of the set of special nilpotent orbits. It would be interesting to know what replacement for the notion of strongly parabolic Higgs bundles would lead us exactly to the set of all special nilpotent orbits. 
\end{remark}

\begin{remark}
We note that some of the underlying mathematics (having to do with Spaltenstein's work on the Kazhdan-Lusztig map \cite{spaltenstein1988polynomials,spaltenstein1990kazhdan} ) appearing in our study of the structure of the local Hitchin image had also appeared in a recent study of an analog of the Hitchin fiber in the local setting \cite{MR4334167}. It is interesting to ask if a direct connection can be found between the features we find in the local Hitchin image and the geometry of the fibers in \cite{MR4334167}.
\end{remark}

\begin{remark}
    While this paper is confined to the case of type-D Hitchin systems, it is natural to wonder if similar features exist for other Cartan types. The connection to class-S theories extends most directly to all Hitchin systems corresponding to simply laced Lie algebras. So, the next natural set of examples to study would be Hitchin systems of type $E_6,E_7,E_8$. These have been studied and the explicit structure of constraints has been worked out for any special nilpotent orbit $\mathcal{O}$ in $E_6$ \cite{Chacaltana:2014jba} and $E_7$ \cite{Chacaltana:2017boe}. In the case of $E_8$, partial results on the constraints have been worked out in unpublished work by the authors of \cite{Chacaltana:2018vhp}. 
    
    As in type-D, the pattern of constraints is closely tied to the size of the dual special piece which, in some cases, is smaller than the size of Lusztig's quotient $\overline{A}(\mathcal{O})$. In the type-D case, we show in Theorem \ref{odd-type-theorem} that a smaller group $\overline{A}_b (\mathcal{O}) \subset \overline{A} (\mathcal{O}) $ captures the size of the dual special piece. Using the case-by-case computations in  \cite{Chacaltana:2014jba,Chacaltana:2017boe,Chacaltana:2018vhp,Chacaltana:2016shw}, we can define such a smaller group $\overline{A}_b (\mathcal{O})$ for each special nilpotent orbit in $E_6,E_7,E_8$.  We believe it should also be possible to give a case-free definition of the $\overline{A}_b$ group in a way that extends naturally to exceptional cases and reduces to the one given in Def \ref{abarb-definition} for the type-D cases. 

    For the non-simply laced cases, the connection to class-S theories is less direct. We do, however, note some possible directions. On the one hand, we have the local results which relate twisted defects to nilpotent orbits in non-simply laced Lie algebras \cite{Chacaltana:2012ch,Chacaltana:2013oka,Chacaltana:2015bna}. On the other hand, we have certain global results in the mathematical literature which relate Hitchin systems (without punctures) of type $B,C,F,G$ to Hitchin systems of type $A,D,E$ \cite{beck2022folding}.  Extending the global results to the punctured case using the local results on twisted defects might be a way to study corresponding questions for the non-simply laced cases.
\end{remark}

\subsection{The local setup}
\label{the_local_setup}

Let $\mathcal{O} = C[[t]]$ be the ring of formal power series and $\mathcal{F} = C((t))$ be the field of formal Laurent series. The Higgs field in the neighbourhood of a puncture has the following form

\begin{displaymath}
\phi' = \frac{\mathfrak{n}}{t} + \text{reg}
\end{displaymath}
where $\mathfrak{n} \in \mathfrak{g}$ is a nilpotent element living in the (co-adjoint) orbit $\mathcal{O}_H$ (ie. the Hitchin nilpotent orbit). Note that the above condition is equivalent to saying that $t \phi'$ is of the form $\mathfrak{n} + m \mathfrak{g}$ where $m$ is the unique maximal ideal in $\mathcal{O}$. This version will be helpful in relating to \cite{kazhdan1988fixed,spaltenstein1988polynomials,MR1079990}. 

Nilpotent orbits in the classical Lie algebras are labeled by partitions. For $\mathfrak{so}(2N)$, we have a D-partition $[p]$, a partition of $2N$ such that every even part occurs with even multiplicity\footnote{When all of the parts of $[p]$ are even, then there are actually \emph{two} nilpotent orbits corresponding to that partition. See \cite{CollingwoodMcGovern}.}. 
In what follows we will restrict to $\mathfrak{n}$ in a \textit{special} nilpotent orbit.

\begin{defn} (Special orbits in type-D)
    A special nilpotent orbit in type-D is a nilpotent orbit $\mathcal{O}_H$ whose partition label $[p]$  is a D-partition (all even parts occur with even multiplicity) such that there is an even number (possibly zero) of odd parts between any even parts and an even number (possibly zero) of odd parts at the beginning of the partition.  Equivalently \cite{CollingwoodMcGovern}, $[p]^T$ is a C-partition\footnote{A C-partition is a partition of $2N$ such that every odd part occurs with even multiplicity.}. \label{def_special_orbit}
\end{defn}

\begin{remark}
A Richardson orbit in type-D may be characterized as follows (\cite{MR0732807} 4.3):
\begin{itemize}
\item $[p]$ is special.
\item Let $r=\min_i\{i|p_i\; \text{even}\}-1$. Then, for all $j\geq \tfrac{1}{2}r+1$, if $p_{2j}=2l+1$ and  $p_{2j+1}=2m+1$, then $l>m$.
\end{itemize}
Evidently, all even orbits (orbits whose \emph{weighted Dynkin diagram} consists of 0's and 2's) are Richardson.
\end{remark}

Now, consider the map

\begin{displaymath}
\mu : \mathfrak{g}(\mathcal{F}) \rightarrow \mathbb{A}^N_{\mathcal{F}}
\end{displaymath}
that is defined as follows

\begin{displaymath}
\mu : x \rightarrow (Q_1(x),Q_2(x),\cdots,Q_N(x))
\end{displaymath}
where $Q_i$ are some chosen generators for the ring of invariant polynomials of $\mathfrak{g}$. We define  the \textit{local Hitchin image} associated to the nilpotent orbit $\mathcal{O}_H$ to be the image $\mu(\mathfrak{n} + m \mathfrak{g})$.

\paragraph*{{Zero Orders and Pole Orders}}\label{zero_orders_and_pole_orders}

We have a choice to make as to where the Higgs field lives. If we choose $\phi' \in H^0(ad(V) \otimes K)$, tame Higgs bundles on the punctured disc correspond to those where the Higgs field has a simple pole ($\phi' = \frac{a}{t} + \text{reg}$) at the puncture and the invariant polynomials built out of the Higgs field have poles of various orders. We denote these pole orders by $\pi_k$ where $k$ runs over the degrees of the invariant polynomials $Q_i$, $i\in \{1,2,\ldots , N \}$. 

On the other hand, if we choose $\phi \in H^0(ad(V) \otimes K(D))$, tame Higgs bundles correspond to those where $\phi \in \mathbb{C}[[t]]$ (with non-zero constant term) and the invariant polynomials built out of the Higgs field have zeros of order $\chi_k$ (in the notation in \cite{Balasubramanian:2020fwc}, where $k$ is the degree of the corresponding $Q_i$). The relationship between the two definitions for the Higgs fields is simply

\begin{displaymath}
\phi = t \phi'
\end{displaymath}
From this, it follows that the zero orders and pole orders are related by :

\begin{displaymath}
\chi_k = k - \pi_k
\end{displaymath}

\subsection{ Newton polygons}
\label{np-section}


It is sometimes convenient to draw Newton polygons to keep track of the non-zero coefficients in the Laurent tails of the coefficients of $\lambda$ in $\charpol(\phi - \lambda I )$.

Let $[p_j]$ be the Hitchin partition and let $[s_j]$ be the partial sums built out of the Hitchin partition where $s_j = \sum_{i=1}^{j}p_i$. Let $\sum_j {p_j}=n$. We follow the conventions of \cite{MR3815160} and define a Newton polygon $NP([p_j])$ associated to a partition $[p_j]$ to be the polygon bounded by line segments joining the points $(j,n-s_j)$. \footnote{We note here that Spaltenstein uses a different convention in \cite{spaltenstein1988polynomials} and uses a Newton polygon bounded by line segements joining the points $(n-s_j,j)$.  }  In Fig \ref{np-example}, we give an example of a Newton polygon for $n=12$ and Hitchin partition $[5,3,2^2]$.

\begin{figure}
\begin{center}
\begin{tikzpicture}[scale=0.5]

\draw[-,blue] (0,12)--(1,7)--(2,4)--(3,2)--(4,0) ;

\path[fill=blue!20] (0,12)--(1,7)--(2,4)--(3,2)--(4,0)--(13,0)--(13,12)--(0,12);

  \draw[->] (0,0) -- (14,0) ;
  \draw[->] (0,0) -- (0,14) ;

  \foreach \y in {0,...,11} {
    \foreach \x in {0} {
      \draw[fill,black] (\x, \y) circle (1.3pt) ;
    }
  }

  \foreach \y in {0,...,6} {
    \foreach \x in {1} {
      \draw[fill,black] (\x, \y) circle (1.3pt) ;
    }
  }

   \foreach \y in {0,...,3} {
    \foreach \x in {2} {
      \draw[fill,black] (\x, \y) circle (1.3pt) ;
    }
  }

   \foreach \y in {0,1} {
    \foreach \x in {3} {
      \draw[fill,black] (\x, \y) circle (1.3pt) ;
    }
  }

  \foreach \y in {12,13} {
    \foreach \x in {0} {
      \draw[fill,blue] (\x, \y) circle (1.3pt) ;
    }
  }

    \foreach \y in {7,...,13} {
    \foreach \x in {1} {
      \draw[fill,blue] (\x, \y) circle (1.3pt) ;
    }
  }

  \foreach \y in {4,...,13} {
    \foreach \x in {2} {
      \draw[fill,blue] (\x, \y) circle (1.3pt) ;
    }
  }

  \foreach \y in {2,...,13} {
    \foreach \x in {3} {
      \draw[fill,blue] (\x, \y) circle (1.3pt) ;
    }
  }

  \foreach \y in {0,...,13} {
    \foreach \x in {4} {
      \draw[fill,blue] (\x, \y) circle (1.3pt) ;
    }
  }

  \foreach \y in {0,...,13} {
    \foreach \x in {5,6,7,8,...,13} {
      \draw[fill,blue] (\x, \y) circle (1.3pt) ;
    }
  }
\end{tikzpicture}
\caption{A Newton polygon showing the non-zero coefficients in NP([$5,3,2^2$]). We have chosen the conventions from \cite{MR3815160} and taken the vertices to be at $(j,n-s_j)$ } 
\label{np-example}
\end{center}
\end{figure}
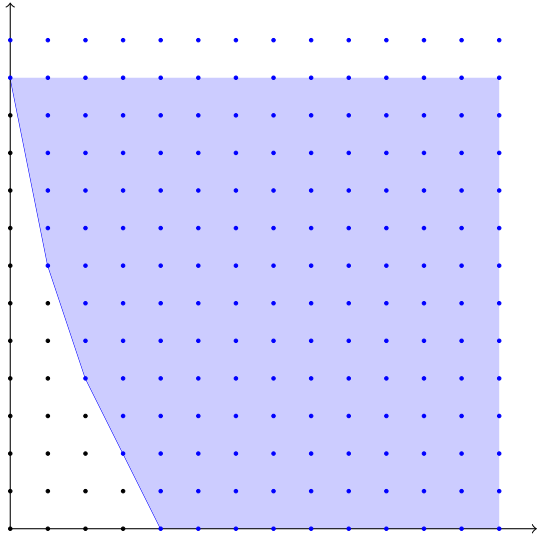

We denote a generic polynomial of degree $n$ with coefficients in $\mathcal{O}$ to be :

\begin{displaymath}
\mathcal{P}_n= \lambda^{n} + \sum _{\beta} \sum_{\alpha} \rho_{\alpha,\beta} t^{\alpha} \lambda^{\beta}
\end{displaymath}
We say that a polynomial $\mathcal{P}_n$ corresponds to a Newton polygon $NP$ if the coefficients $\rho_{\alpha,\beta}$ are non-zero only if $(\alpha,\beta)$ lie in the region defined by the polygon (see shaded region in example above). Note that, by definition, we always have $\rho_{(0,n)}=1$.

We denote by $NP([\mu_i])$ a Newton polygon determined by a partition $[\mu_i]$. Now, let $\charpol(\phi)$ in type $D_N$ be :

\begin{equation}\label{charpolD}
\charpol(\phi) = \lambda^{2N} + u_2\lambda^{2N-2} + \cdots + u_{2N-2}\lambda^{2} + u_{N}^2
\end{equation}
with coefficients $u_{2j},u_N \in \mathcal{O}[[t]]$, the ring of formal power series. Let $j \in 1,2,\ldots,N-1$.  We will denote the leading terms, as $t\to 0$, of $u_{2j}[[t]],u_N[[t]]$  as $c_{2j}t^{\chi_{2j}},\tilde{c}t^{\chi_N}$ respectively. The leading terms have zero orders $\chi_{2j},\chi_N$ and $c_{2j},\tilde{c}$ are the coefficients.  A basic result that determines the zero orders is the following :

\begin{lemma}
\label{lemma1}
   The Newton polygon corresponding to $\charpol(\phi)$ is $ NP([p_j])$ where $[p_j]$ is the Hitchin partition.
\end{lemma}

\begin{proof} The Lemma follows directly from Cor 23 in \cite{MR3815160}.
\end{proof}

The space of $u_{2j}[[t]],u_N[[t]]$ determines what we called $\mathcal{B}^{\text{naive}}$. In type $D$, there are some additional constraints obeyed by the leading terms $c_{2j},\tilde{c}$. These additional constraints are the main subject matter of this paper. In subsequent sections, we will turn to studying these constraints from a local point of view. 

\subsection{Spaltenstein's factorization}\label{sec:Spaltenstein}

Our analysis will heavily rely on results  of Spaltenstein on factorization patterns of characteristic polynomials and the Kazhdan-Lusztig map. In particular, we need Theorem B and Proposition 6.4 from \cite{spaltenstein1988polynomials}. We specialize the theorem to the case of Higgs fields with residues in special orbits.

\begin{theorem} (Spaltenstein \cite{spaltenstein1988polynomials})
For any special nilpotent orbit $\mathcal{O}_H$ in $\mathfrak{so}_{2N}$, let $\charpol( \phi)$ be the characteristic polynomial built out of the Higgs field $\phi=t\phi'$ with local behaviour $\phi' = \mathfrak{n}/t + (reg)$ and $\mathfrak{n}$ being an element in the orbit $\mathcal{O}_H$. Then  $\charpol( \phi)$ has a factorization
\begin{equation}\label{eq:spaltfact}
\charpol(\phi)= \prod_i P_{\alpha,i}(\lambda)P_{\alpha,i}(-\lambda)\prod_j P_{\beta,j}(\lambda^2 ) 
\end{equation}
into irreducible polynomial factors labeled by a bi-partition $(\alpha,\beta)$ of $N$.
\label{spaltenstein-big-theorem}
\end{theorem}

First, the bi-partition $(\alpha,\beta)$ is determined from the special partition $[p]$ in the following way

\begin{itemize}%
\item When $p_{2j-1}=p_{2j}$, then $\alpha$ has one part equal to $p_{2j-1}$.
\item When $p_{2j-1}=2r+1$ and $p_{2j}=2s+1$, with $r>s$, then $\beta$ has two parts: $r$ and $s+1$.
\end{itemize}
Then, the irreducible factors are determined from the parts of $(\alpha,\beta)$ according to :

\begin{itemize}%
\item Corresponding to every part $\alpha_i$ of $[\alpha]$, there is a pair of irreducible factors: $P_{\alpha,i}(\lambda), P_{\alpha,i}(-\lambda)$, of $\charpol( \phi)$ each of degree $\alpha_i$ where
\[P_{\alpha,i}(\lambda) := (\lambda^{\alpha_i} + t e_1^{(i)}[[t]]\lambda^{{\alpha_i}-1} +  t e_2^{(i)}[[t]]\lambda^{{\alpha_i}-2} +  \cdots + t e_{{\alpha_i}}^{(i)}[[t]]).\]

\item Corresponding to every part $\beta_i$ of $\beta$, there is a single irreducible factor, {\it i.e.} a polynomial $P_{\beta,i}(\lambda^2)$, of degree $\beta_i$. 
\[P_{\beta,i}(\lambda^2 ) := (\lambda^{2{\beta_i}} +  t f_2^{(i)}[[t]] \lambda^{2{\beta_i} -2} + t f_4^{(i)}[[t]] \lambda^{2{\beta_i} -4} +\cdots + t f_{2{\beta_i}}^{(i)}[[t]]  ).\]

\end{itemize}

The coefficients $e_n[[t]]$ and $f_n[[t]]$ are formal power series in $t$. The coefficients in those power series are, in general, algebraic (i.e., not necessarily single-valued) functions on the Lie algebra. Of course, $\charpol( \phi)$ is single-valued. In fact, its coefficients (as a double power series in $\lambda$ and $t$) are polynomials on the Lie algebra. Thus, whatever monodromies are experienced by the coefficient functions $e_n[[t]]$ and $f_n[[t]]$ must be such as to exchange the polynomial factors in \eqref{eq:spaltfact}. The condition that the partition be special implies that no $\alpha_i$ can equal $ 2 \beta_j$, so no $P_{\alpha}(\lambda)$ can have the same degree as a $P_{\beta}(\lambda^2)$. In particular, monodromy permutes the $P_{\alpha}$'s and $e$'s among themselves, and similarly for the $P_{\beta}$'s and $f$'s, with no mixing.

The parts of $[\beta]$ can only occur with multiplicity 2 or 1. For $[p]$ special, the multiplicity is 2 iff $p_{2j-1}=p_{2j}+2=2r+1$, in which case  $\beta_i=\beta_{i+1}=r$, for some $i$. In this case, it is the product $P_{\beta,i}(\lambda^2 )P_{\beta,i+1}(\lambda^2 )$ that is carried into itself under monodromy. Hence (in particular) $f^{(i)}_{2r}f^{(i+1)}_{2r}$ is a rational function on the Lie algebra.

When $p_{2j-1}=p_{2j}+2l=2r+1$ (for some $l>1$), then $P_{\beta,i}(\lambda^2)$ and $P_{\beta,i+1}(\lambda^2)$ are separately mondromy-invariant and hence $f^{(i)}_{2r}[[t]]$ and $f^{(i+1)}_{2(r-l+1)}[[t]]$ are separately rational functions on the Lie algebra.

If the parts of $[\alpha]$ are all distinct, then the monodromy can \emph{at worst} exchange $P_{\alpha,i}(\lambda)\leftrightarrow P_{\alpha,i}(-\lambda)$. In other words, it could act on the $e_n[[t]]$ as $e_n[[t]]\to (-1)^n e_n[[t]]$. So the (coefficients of powers of $t$ in) $e_{2n}[[t]]$
are rational functions on the Lie algebra, whereas the coefficients in  $e_{2n+1}[[t]]$ can have the form of a rational function times the square-root of a polynomial.

When the parts of $[\alpha]$ appear with multiplicity greater than 1,  then more complicated monodromy patterns can appear. Let $[p]=[\dots,(2r)^{2s},\dots]$, and let the corresponding parts of $[\alpha]$ be $\alpha_i=\alpha_{i+1}=\dots=\alpha_{i+s-1}=2r$. Then Spaltenstein's factorization contains a product of the form
\[
\prod_{j=0}^{s-1} P_{\alpha,{i+j}}(\lambda) P_{\alpha,{i+j}}(-\lambda)
\]
of polynomials, each of which has degree $2r$. These can have monodromies among themselves, but the product
\[
  \prod_{j=0}^{s-1} e_{2r}^{(i+j)}[[t]]
\]
is monodromy-invariant, and hence a rational function on the Lie algebra. To see this, simply note that $Q_{2r}(\lambda)=\prod_{j=0}^{s-1} P_{\alpha_{i-j}}(\lambda)$ can only be exchanged with $Q_{2r}(-\lambda)$ under monodromy and that $\prod_{j=0}^{s-1} e_{2r}^{(i+j)}[[t]]$ is the coefficient of $\lambda^0$ in both $Q_{2r}(\lambda)$ and $Q_{2r}(-\lambda)$.

Similarly, if $[p]=[\dots,(2r+1)^{4s},\dots]$, then Spaltenstein's factorization contains a product of the form
\[
\prod_{j=0}^{2s-1} P_{\alpha,{i+j}}(\lambda) P_{\alpha,{i+j}}(-\lambda)
\]
of polynomials of degree $2r+1$. Again, these polynomials can have monodromies among themselves, but the product
\[
  \prod_{j=0}^{2s-1} e_{2r+1}^{(i+j)}[[t]]
\]
is monodromy-invariant, and hence a rational function on the Lie algebra. 

\noindent {\bf Summary:} The following products of coefficients are rational functions on the Lie algebra:
\begin{itemize}
\item   $f^{(2i-1)}_{2r}[[t]]$ and $f^{(2i)}_{2(r-l+1)}[[t]]$ when $p_{2j-1}=p_{2j}+2l=2r+1$ (for some $l>1$)
\item $f^{(2i-1)}_{2r}f^{(2i)}_{2r}$ when $p_{2j-1}=p_{2j}+2=2r+1$.
\item $e_{2n}[[t]]$ if the parts of $[\alpha]$ are all distinct.
\item $\prod_{j=0}^{s-1} e_{2r}^{(i+j)}[[t]]$ if $[p]=[\dots,(2r)^{2s},\dots]$.
\item   $\prod_{j=0}^{2s-1} e_{2r+1}^{(i+j)}[[t]]$ if 
$[p]=[\dots,(2r+1)^{4s},\dots]$.
\end{itemize}

The conjecture of the next section is, in essence, the assertion that certain monodromy patterns do not, in fact, occur and that certain additional coefficients (or products of coefficients) in the $e^{(i)}_n[[t]]$ and $f^{(i)}_{2n}[[t]]$ are rational functions on the Lie algebra, even though the factorization \eqref{eq:spaltfact} would allow them to have nontrivial monodromies.

\section{A Conjecture}\label{sec:main_conjecture}
Given a special D-partition, $[p]$ and a nilpotent, $X$, in the corresponding nilpotent orbit, let us write the characteristic polynomial $\charpol(\phi)=\det(X+tg-\lambda I)$, $g\in\mathfrak{so}(2N)$ as
\begin{equation}\label{charpolexp}
\charpol(\phi)
=\lambda^{2N}+\sum_{k=1}^N t^{\chi_{2k}} \bigl(c_{2k}+O(t)\bigr)\lambda^{2(N-k)}
\end{equation}
 Here, if $[p]$ has parts $p_i$ then
\[
\chi_{2k}=j,\qquad \text{where}\quad \sum_{i=1}^{j-1}p_i < 2k\leq \sum_{i=1}^j p_i
\]
A-priori, the $c_{2k}$ are polynomials on the Lie algebra, of degree $\chi_{2k}$. In certain cases, a stronger result obtains.
\begin{conjecture}\label{Simpleconjecture}
Let $[p]$ be a special D-partition. For every part $p_{2i}$ such that $p_{2i}>p_{2i+1}$,  let $2k_i=\sum_{j=1}^{2i}p_j$. Then $c_{2k_i}=(a_{k_i})^2$ with $a_{k_i}$ a polynomial on the Lie algebra.
\end{conjecture}
An obvious special case is when $p_{2i}$ is the last part of $[p]$, {\it i.e.~} when $k=N$. Then $c_{2N}$ is the square of the Pfaffian.  More generally, we have:

\begin{prop}\label{mult4}
Conjecture \ref{Simpleconjecture} holds in the following cases:
\begin{itemize}
\item[a)] If $p_1,p_2,\dots, p_{2i}$ are all even then the Conjecture holds for $c_{2k_i}$.
\item[b)] If all of the parts $p_j$, for $j>2i$ are even, then the Conjecture hold for $c_{2k_i}$.
\item[c)] If all odd parts of $[p]$ occur with multiplicity a multiple of 4, then Conjecture \ref{Simpleconjecture} holds, with $a_{k_i} =\prod_{j=1}^{i} e^{(j)}_{p_{2j}}$. 
\end{itemize}
\end{prop}

\begin{proof}
Part (a) follows from the proof of part (i) of Theorem \ref{th:BKCD} below. Part (b) follows from the proof of part (iii) of Theorem \ref{th:BKCD}. For Part (c), the individual $e^{(j)}_{p_{2j}}$ are not necessarily rational functions on the Lie algebra. Nevertheless, for each even part of $p$ (occurring with multiplicity $2s$) and for each odd part (occurring with multiplicity $4s$), we showed at the end of Section \ref{sec:Spaltenstein}  that the \emph{product} of the corresponding $e^{(j)}_{p_{2j}}$  is rational. Taking the product of these rational functions, we obtain a rational function whose \emph{square} is a polynomial on the Lie algebra. Hence it, itself, is a polynomial on the Lie algebra.
\end{proof}

We have not proven the Conjecture for general $[p]$ with odd parts whose multiplicity is $4s+2$ or $2s+1$. Nevertheless, Chacaltana-Distler (\cite{Chacaltana:2011ze} and unpublished) checked the Conjecture explicitly for all nilpotent orbits in $\mathfrak{so}(2N)$ for $N\leq 7$.  We have spot-checked a large number of cases for higher $N$.

%

\section{Even and Very Even Type Constraints}
\label{sec:even_and_very_even}

\subsection{Work of Baraglia-Kamgarpour}
\label{sec:spaltenstein}

Now, we apply Spaltenstein's result (Theorem \ref{spaltenstein-big-theorem}) to study the Hitchin image in the presence of even parts in the Hitchin partition. This discussion will closely follow the work of \cite{MR3815160}. Later, we will emphasize how the approach of \cite{Chacaltana:2011ze} differs from that of \cite{MR3815160}. Let $[p]$ be the Hitchin partition. For $[p]$ a special orbit in type $D$, Theorem \ref{spaltenstein-big-theorem} implies that
\begin{displaymath}
\charpol(\phi) = \mathcal{P}^{odd} \prod_{i|\alpha_i\;\text{even}} P_{\alpha,i}(\lambda)P_{\alpha,i}(-\lambda)
\end{displaymath}
where $\mathcal{P}^{odd}$ contains all of the factors in \eqref{eq:spaltfact} corresponding to the odd parts of $[p]$. An even part $[p]=[\dots,(2r)^{2s},\dots]$ corresponds to $[\alpha]=[\dots,(2r)^s,\dots]$, and the product is over the even parts of $[\alpha]$.



Now, following \cite{MR3815160}, for each \emph{distinct} even part $r_d$ of $[p]$, occurring with multiplicity $2e_d$,  let us \emph{define} a polynomial $Q$ that is built out of the leading terms corresponding to the points on the line segments of slope $-r_d$, in the Newton polygon $NP([r_d])$. We will have $2e_d+1$ points on this line segment. Let their co-ordinates be $\alpha_s,\beta_s$ for $s=0,\ldots 2e_d$.

\begin{displaymath}
Q_d(u) =  \sum_{s=0}^{2e_d} \rho_{\alpha_s,\beta_s} u^{2e_d - s}
\end{displaymath}
Note that when the first part in the Hitchin partition is even, the first term in the polynomial is just $u^{2e_d}$ since $\rho_{(0,N)}=1$ by definition.

For special nilpotents, $p_{2i-1},p_{2i}$ have the same parity. In other words, even and odd parts occur in pairs. As derived in the proof of Prop 34 in \cite{MR3815160}, this implies that the polynomial $Q_d$ is a perfect square. That is, we have an equation of the form

\begin{equation}
\sum_{s=0}^{2e_d} \rho_{\alpha_s,\beta_s} u^{2e_d - s} = {\Bigl(\sum_{n=0}^{e_d} \hat{a}_n u^{e_d-n}\Bigr)}^2
\end{equation}

Observe that the coefficients on the LHS involve just the leading terms in the $t$ expansion of the coefficients in $\charpol(\phi)$. So, we could rewrite the above equation using notation introduced at the end of Section \ref{np-section}.
\begin{equation}
\label{asquare}
    \sum_{s=0}^{2e_d} c_{2k +s r_d }u^{2e_d - s} = {\Bigl(\sum_{n=0}^{e_d} a_{k+n r_d} u^{e_d-n}\Bigr)}^2
\end{equation}
where $2k$ is the following partial sum of the Hitchin partition $[p_j]$ 
\begin{equation}
    2k = \sum_{j\;|\; p_j > r_d} p_j
 \end{equation}
and $\hat{a}_n(t)= a_{k+n r_d}t^{\frac{1}{2}\chi_{2k+2n r_d}}+\dots$ for some $t$-independent $a_{k+n r_d}$.

Since there are fewer independent coefficients on the polynomial on the RHS, this implies relations among the $c_{2k+s r_d }$ on the LHS (called c-type constraints in CD).
We will henceforth use the notation $c_i$ to refer to these leading coefficients where the $i$ subscript encodes the $\mathbb{C}^\star$ weight of the coefficient in the graded decomposition of the local Hitchin image. 

We can use \eqref{asquare} to \emph{solve} for as many of the $a_n$s as we can. Here is where CD and BK part company. BK eliminate the remaining (unsolved-for) $a_n$s, whereas CD treat them as independent gauge-invariant polynomials in the coefficients in the Higgs field, hence as generators of the local Hitchin base.

\subsection{Relationship between Chacaltana-Distler and Baraglia-Kamgarpour}
\label{sec:relBKCD}

We summarise the relationship between CD's description of the local Hitchin image and BK's description in the following theorem.

\begin{defn}
    Given a Hitchin partition $[p]$, decompose it into blocks, with one block for
    \begin{itemize}
        \item each odd part (counted with multiplicity)
        \item each string of consecutive even parts
    \end{itemize}
    For both BK and CD, the image of the local Hitchin maps is a Cartesian product of factors, one factor for each block in the above decomposition.
    \begin{equation}\label{blockdecomp}S_{BK}=\prod_i S_{BK}^{(i)},\qquad S_{CD}=\prod_i S_{CD}^{(i)}\end{equation}
    The $S^{(i)}$ will be referred to as the \emph{factors} in the local Hitchin image.
\end{defn}

Consider the factor $S^{(i)}$ corresponding to a particular string of consecutive even parts in $[p]$.
\begin{defn}
(BK's Hitchin image) BK describe $S_{BK}^{(i)}$ as the space of $\rho_{\alpha_s,\beta_s}$ subject to the constraints arising from equation \eqref{asquare}. It describes a (possibly singular) subspace of co-dimension $e_j/2$ in $\mathbb{C}^{e_j+1}$. 
\end{defn}

\begin{defn}
 (CD's Hitchin image) CD's description, $S_{CD}^{(i)}$, amounts to taking the space of $a_n$ as the local Hitchin image. It is manifestly a vector space of dimension $e_j/2+1$.

\end{defn}
\begin{remark}
        The map $\gamma: S_{CD}\to S_{BK}$ respects the decomposition \eqref{blockdecomp}. That is, it consists of a collection of maps $\gamma^{(i)}: S_{CD}^{(i)}\to S_{BK}^{(i)}$.
\end{remark}

\begin{theorem}\label{th:BKCD}
Let $\gamma :S_{CD} \rightarrow S_{BK}$ be the map between the two descriptions of the local Hitchin image and let $\gamma^{(i)}: S_{CD}^{(i)}\to S_{BK}^{(i)}$ be the component corresponding to a particular string of consecutive even parts in $[p]$.
\begin{itemize}%
\item[i)] If the block occurs at the beginning of the partition and consists of either a single even part (with even multiplicity) or is a very-even partition of the form $[(2s)^{2l},(2r)^2]$, then the map $\gamma^{(1)}$ is an isomorphism. Both $S_{CD}^{(1)}$ and $S_{BK}^{(1)}$ are smooth (in fact, vector spaces). For any other string of even parts at the beginning of the partition,  $S_{BK}^{(1)}$ is singular and $\gamma^{(1)}$ is the normalization map. In particular, $S_{CD}^{(1)}$ and $S_{BK}^{(1)}$ are birational.
\item[ii)]When the block occurs at the end of the partition (but does not extend to the beginning of the partition), 
then $S_{BK}^{(i)}$ is singular and $\gamma^{(i)}$ is the normalization map. In particular, $S_{CD}^{(i)}$ and $S_{BK}^{(i)}$ are birational.
\item[iii)] When the block occurs in the middle of the partition (i.e.~not at the beginning or the end) of the partition, then Conjecture \ref{Simpleconjecture} implies that $\gamma^{(i)}$ is generically $2:1$. $S_{CD}^{(i)}$ is a vector space, $r: S_{CD}^{(i)}\circlearrowleft$ is reflection through the origin and $S_{CD}^{(i)}/r$ is birational to $S_{BK}^{(i)}$.
\end{itemize}
\label{CD-BK-theorem}
\end{theorem}

Before turning to the proof of \textbf {Theorem \ref{th:BKCD}}, let us look at a few examples.

\begin{enumerate}
\item Consider the partition $[4^2,1^4]$. The string of even parts is at the beginning of the partition. Consequently, $Q_1(u)$ has a leading term with coefficient 1.
\begin{displaymath}
u^2 + c_4 u +c_8 = {(a_0 u+ a_4)}^2
\end{displaymath}
This can be solved by setting $a_0=1$, $a_4 = c_4/2$. Then we have a \emph{constraint}, $c_8 = \tfrac{1}{4} c_4^2$. Eliminating $c_8$, CD and BK agree: $S_{\text{CD}}=S_{\text{BK}}$

\item Consider the partition $[2^6]$. This is a very-even partition, so it consists of a single block of even parts. Since it starts with an even part,  $Q_1(u)$ has a leading term with coefficient 1.

\begin{displaymath}
u^6+c_2 u^5+c_4 u^4+c_6 u^3 +c_8 u^2+c_{10} u + \tilde{c}^2 = {(a_0 u^3+a_2 u^2+a_4 u +a_6)}^2
\end{displaymath}
We set $a_0=1$ and solve for $a_2=c_2/2$. The subsequent terms yield

\begin{displaymath}
\begin{aligned}
a_2^2 +2a_0 a_4 &= c_4\\
2(a_0 a_6 +a_2 a_4) &= c_6\\
2 a_2 a_6 + a_4^2 &= c_8\\
a_6&=\pm \tilde{c}
\end{aligned}
\end{displaymath}
which can be successively solved for $a_{4,6}$:

\begin{displaymath}
\begin{aligned}
a_4&=\tfrac{1}{2}\left(c_4-\tfrac{1}{4}c_2^2\right)\\
a_6&= \tfrac{1}{2}c_6-\tfrac{1}{4} c_2(c_4-\tfrac{1}{4}c_2^2)
\end{aligned}
\end{displaymath}
Plugging back in yields the constraints

\begin{equation}
\begin{aligned}
   c_6\mp2\tilde{c}&=\tfrac{1}{2}c_2(c_4-\tfrac{1}{4} c_2^2)\\
   c_8 &=\tfrac{1}{4}{\left(c_4-\tfrac{1}{4}c_2^2\right)}^2\pm c_2\tilde{c}\\
c_{10}&=\pm\tilde{c}(c_4-\tfrac{1}{4}c_2^2)
\end{aligned}
\label{222222constraints}\end{equation}
The two solutions for $a_6$, which yield the two choices of sign in \eqref{222222constraints}, correspond to the two nilpotent orbits associated to this very-even partition. Again, CD and BK agree.

\item Consider the partition $[4^2,2^2]$. Again, the partition is very-even, so there is a single block of even parts. However, there are two distinct even parts, so we get two perfect squares. Again, the leading part is even, so $Q_1(u)$ has leading coefficient 1.

\begin{displaymath}
\begin{aligned}
   u^2+c_4 u +c_8 &= (a_0 u+ a_4)^2\\
   c_8 u^2 +c_{10} u + \tilde{c}^2 &= (a_4 u +a_6)^2
  \end{aligned}
\end{displaymath}
Here, we set $a_0=1$ and $a_4=c_4/2$, which yields the constraint $c_8 =\tfrac{1}{4} c_4^2$. Plugging that into the second equation, we have $a_4=c_4/2$ and $a_6= \pm \tilde{c}$. This yields the constraint $c_{10}=\pm \tilde{c}c_4$. We thus have two solutions, corresponding to the two nilpotent orbits determined by this very-even partition.

\item Consider the very-even partition $[4^2,2^4]$. We again get two perfect squares
\[
\begin{split}
u^2+c_4 u+c_8&= (u+a_4)^2\\
c_8 u^4+c_{10}u^3+c_{12}u^2+c_{14}u+\tilde{c}^2&=(a_4 u^2+a_6 u +a_8)^2
\end{split}
\]
Both BK and CD can solve for $a_4=\tfrac{1}{2}c_4$ and thence for $c_8$
\[
c_8=\tfrac{1}{4} c_4^2
\]
From the second equation, we have two possible solutions for $a_8$, $a_8=\pm\tilde{c}$, corresponding to the two choice of nilpotent orbit determined by this very-even partition. The remaining constraint equations read
\begin{equation}\label{part442222}
\begin{split}
c_{10}&= c_4 a_6\\
c_{12}&= a_6^2\pm c_4\tilde{c}\\
c_{14}&=\pm 2\tilde{c}a_6
\end{split}
\end{equation}
CD take the local Hitchin base to be the vector space spanned\footnote{$c_2,c_6$ don't participate in the constraints, so CD and BK agree that they form a vector subspace of the local Hitchin base.} by $c_2,c_4,c_6,a_6,\tilde{c}$. BK eliminate $a_6$ from \eqref{part442222} and write the local Hitchin base as the singular affine variety
\[
S_{\text{BK}}=\left\{\begin{aligned}
c_{14}^2 - 4  c_{12} \tilde{c}^2 \pm 4 c_{4} \tilde{c}^3&=0\\
c_{10} c_{14} \mp 2 c_{4}  c_{12} \tilde{c} + 2 c_{4}^2 \tilde{c}^2&=0\\
c_{10}^2 - c_{4}^2 c_{12} \pm  c_{4}^3 \tilde{c}&=0\\
c_{4} c_{14} \mp 2 c_{10} \tilde{c}&=0
\end{aligned}\right\}\subset\mathbb{C}^7
\]
Away from the singular locus
\[
\{c_4=c_{10}=c_{14}^2-4c_{12}\tilde{c}^2=0\}
\]
the map $\gamma: S_{\text{CD}}\to S_{\text{BK}}$ is 1-1. On the singular locus, it is generically 2-1 (except on the line $\{c_4=c_{10}=c_{12}=c_{14}=0\}$ where it is 1-1), hence a normalization.

\item Consider the partition $[5,3,2^2]$. This partition has the block of even parts at the end. Hence CD and BK disagree. We have 

\begin{displaymath}
c_8 u^2+ c_{10} u + \tilde{c}^2 = (a_4 u +a_6)^2
\end{displaymath}
We can solve for $a_6 =\tilde{c}$. But the remaining constraints then read

\begin{equation}
c_8= a_4^2 ,\qquad c_{10} = 2a_4 \tilde{c}
\label{CDBKsol}\end{equation}
CD say that $a_4$ and $\tilde{c}$ parametrize the Hitchin image (instead of $c_8, c_{10}$ and $\tilde{c}$). BK eliminate $a_4$ and say this factor in the Hitchin image is parametrized by $c_8, c_{10}$ and $\tilde{c}$, with the constraint $c_{10}^2=4c_8\tilde{c}^2$.

The two spaces, $\mathcal{S}_{\text{CD}}$ and $\mathcal{S}_{\text{BK}}$ are birational. The map $\mathcal{S}_{\text{CD}}\to \mathcal{S}_{\text{BK}}$ is 1:1 away from the locus $\{\tilde{c}=0,\; c_8\neq 0\}$, where it is 2:1. Again, $\mathcal{S}_{\text{CD}}$ is the normalization of $\mathcal{S}_{\text{BK}}$ 

\item Consider the partition $[5^2,4^2,2^4,1^2]$. We again get two perfect squares
\[
\begin{split}
c_{10} u^2 + c_{14} u + c_{18} &= (a_5 u+ a_9)^2\\
c_{18} u^4 + c_{20} u^3+c_{22}u^2+c_{24}u+c_{26}&=(a_9 u^2+a_{11} u +a_{13})^2
\end{split}
\]
In this case equating coefficients in the polynomial equations yields:
\[
\begin{aligned}
    c_{10} &= a_5^2 \\
    c_{14} &= 2 a_5 a_9 \\
    c_{18} &= a_9^2 \\
    c_{20} &= 2 a_9 a_{11} \\
    c_{22} &= a_{11}^2 + 2 a_9 a_{13} \\
    c_{24} &= 2 a_{11} a_{13} \\
    c_{26} &= a_{13}^2 
\end{aligned}
\]
Here BK would eliminate $a_5$, $a_9$, $a_{11}$ and $a_{13} $ and write this factor in the local Hitchin base, $S^{(i)}$, as the singular affine variety which is the intersection of a quadric and seven cubics in $\mathbb{C}^9$ (spanned by $c_{10},c_{12},\dots, c_{26}$).
\begin{equation}\label{CarlosExample}
 S=\left\{
\begin{aligned}
c_{14}^2 - 4 c_{10} c_{18}&=0\\ 
 c_{20}^3 - 4c_{18} c_{20} c_{22} + 8c_{18}^2 c_{24}&=0\\
 c_{20}^2 c_{22} - 4 c_{18} c_{22}^2 + 2 c_{18} c_{20} c_{24} + 16 c_{18}^2 c_{26}&=0\\
 c_{20}^2 c_{24} - 4 c_{18} c_{22} c_{24} + 8 c_{18} c_{20} c_{26}&=0\\
 c_{18} c_{24}^2 - c_{20}^2 c_{26}&=0\\
 c_{20} c_{24}^2 - 4 c_{20} c_{22} c_{26} + 8 c_{18} c_{24} c_{26}&=0\\
 c_{22} c_{24}^2 - 4 c_{22}^2 c_{26} + 2 c_{20} c_{24} c_{26} + 16 c_{18} c_{26}^2&=0\\
 c_{24}^3 - 4 c_{22} c_{24} c_{26} + 8 c_{20} c_{26}^2&=0
 \end{aligned}
\right\}\subset\mathbb{C}^9   
\end{equation}
while CD would take $S^{(i)}$ to be the vector space $\mathbb{C}^6$ spanned by $a_5,a_{9},a_{11},a_{13},c_{12},c_{16} $. The map $\gamma^{(i)}:S_{\text{CD}}^{(i)}\to S_{\text{BK}}^{(i)}$, is 2-1 away from the singular locus
\[
\{c_{14}=c_{18}=c_{20}=c_{24}^2-4c_{22}c_{26}=0\}
\]
 of \eqref{CarlosExample}, where it is \hbox{finite-1}. If $r:\{a_5,a_{9},a_{11},a_{13}\}\mapsto \{-a_5,-a_{9},-a_{11},-a_{13}\}$, then $S^{(i)}_{\text{CD}}/r$ is birational to $S^{(i)}_{\text{BK}}$.

\item Finally, let us consider a partition with three distinct blocks of even parts, one at the beginning, one in the middle and one at the end:  $[6^2,5^2,4^4,3^2,2^2]$ in $D_{24}$. We have

\begin{displaymath}
\begin{split}
  u^2+ c_6 u+ c_{12} &= (a_0 u+a_6)^2\\
  c_{22}u^4+c_{26}u^3+c_{30}u^2+c_{34}u+c_{38} &= (a_{11} u^2+a_{15} u +a_{19})^2\\
  c_{44}u^2+c_{46}u+ \tilde{c}^2 &= (a_{22} u +a_{24})^2
 \end{split}
\end{displaymath}
The sign ambiguity among the $a$'s in the first equation can be resolved by setting $a_0=1$. Then we have $2a_6=c_6$, $a_6^2=c_{12}$. Or

\begin{displaymath}
c_{12} = \tfrac{1}{4} c_6^2.
\end{displaymath}
On this, CD and BK agree.

Similarly, among $a_{22},a_{24}$, the sign ambiguity can be fixed by setting $a_{24}=\tilde{c}$. Then we have $2a_{22}\tilde{c}=c_{46}$, $a_{22}^2=c_{44}$.

BK would eliminate $a_{22}$ from these equations and write the singular hypersurface in $\mathbb{C}^3$ (spanned by $c_{44},c_{46},\tilde{c}$)

\begin{displaymath}
4 c_{44} \tilde{c}^2 = c_{46}^2
\end{displaymath}
CD would construct the vector space, $\mathbb{C}^2$, spanned by $a_{22}$ and $\tilde{c}$. The map to BK's space is given by

\begin{displaymath}
\begin{split}
    c_{44}&=a_{22}^2\\
    c_{46}&=2a_{22}\tilde{c}\\
    \tilde{c}&=\tilde{c}
 \end{split}
\end{displaymath}
This map is 1-1 away from the locus $\{\tilde{c}=0,\; c_{44}\neq0\}$. On that locus, it is 2-1.

 Finally, let us turn to the block of even parts in the middle. Both CD and BK agree on the vector space factor, spanned by $c_{24},c_{28},c_{32},c_{36}$. The constraints on the remaining $c$'s yield
\begin{displaymath}
\begin{split}
c_{22}&= a_{11}^2\\
c_{26}&= 2a_{11} a_{15}\\
c_{30}&= 2a_{11} a_{19}+a_{15}^2\\
c_{34}&=2a_{15}a_{19}\\
c_{38}&= a_{19}^2
\end{split}
\end{displaymath}
BK eliminate $a_{11},a_{15},a_{19}$ and write $S\subset\mathbb{C}^5$ as the singular intersection of 7 cubics:

\begin{displaymath}
S=\left\{
\begin{aligned}
 c_{22}^3 - 4c_{26} c_{22} c_{30} + 8c_{26}^2 c_{34}&=0\\
 c_{22}^2 c_{30} - 4 c_{26} c_{30}^2 + 2 c_{26} c_{22} c_{34} + 16 c_{26}^2 c_{38}&=0\\
 c_{22}^2 c_{34} - 4 c_{26} c_{30} c_{34} + 8 c_{26} c_{22} c_{38}&=0\\
 c_{26} c_{34}^2 - c_{22}^2 c_{38}&=0\\
 c_{22} c_{34}^2 - 4 c_{22} c_{30} c_{38} + 8 c_{26} c_{34} c_{38}&=0\\
 c_{30} c_{34}^2 - 4 c_{30}^2 c_{38} + 2 c_{22} c_{34} c_{38} + 16 c_{26} c_{38}^2&=0\\
 c_{34}^3 - 4 c_{30} c_{34} c_{38} + 8 c_{22} c_{38}^2&=0
\end{aligned}
\right\}\subset\mathbb{C}^5
\end{displaymath}
CD take $S$ to be the vector space $\mathbb{C}^3$, spanned by $a_{11},a_{15},a_{19}$.
\end{enumerate}

\begin{proof}[Proof of Theorem 2]

For part (i) of the Theorem, consider the case where the string of consecutive even parts occurs at the beginning of the partition $[p]$. Say the first part (with multiplicity) is $(2s)^{2l}$. The associated constraint equation reads
\[
u^{2l}+c_{2s}u^{2l-1}+c_{4s}u^{2l-2}+\dots+c_{4ls}=(u^l+a_{2s}u^{l-1}+a_{4s}u^{l-2}+\dots+a_{2ls})^2
\]
The equation for $a_{2s}$ is linear: $a_{2s}=\tfrac{1}{2}c_{2s}$. Plugging its solution into the equation for $a_{4s}$, then becomes linear and can be solved. Proceeding in the same fashion, we  solve for $a_{2s},\dots, a_{2ls}$ which, in turn, allows us to solve for $c_{2s(l+1)},\dots c_{4ls}$ as polynomials in the $c_{2s},\dots, c_{2ls}$.
If this is the only even part in the block, then we are done. If this is a very-even partition of the form $[(2s)^{2l},(2r)^2]$, then the second constraint equation
\[
a_{2ls}^2 u^2 +c_{4ls+2r} u +\tilde{c}^2 = (a_{2ls}u+a_{2ls+2r})^2
\]
which is solved by setting $a_{2ls+2r}=\pm \tilde{c}$ and $c_{4ls+2r}= \pm 2\tilde{c} a_{2ls}$ (recall that we have already solved for $a_{2ls}$). The choice of sign corresponds to the two nilpotent orbits associated to this very-even partition. 

In either of these two cases, we have solved for the $a$'s in terms of the $c$'s and $S^{(1)}_{\text{CD}}\simeq S^{(1)}_{\text{BK}}$. In any other case of a string of even parts at the beginning of the partition, we can continue to solve for the $a$'s, as we have done. However, we only obtain them as rational functions of the $c$'s. For concreteness, suppose $[p]=[(2s)^{2l},(2r)^{2m},\dots]$. We have already solved for $a_{2ls}$ as a polynomial in the $c$'s. The constraint associated to the second distinct even part in this block is
\[
a_{2ls}^2 u^{2m} +c_{4ls+2r} u^{2m-1}+\dots+c_{4ls+4mr}=(a_{2ls} u^m+a_{2ls+2r}u^{m-1}+\dots +a_{2ls+2mr})^2
\]
This is a linear equation for $a_{2ls+2r}$, which is solved by setting $a_{2ls+2r}=\tfrac{c_{4ls+2r}}{2a_{2ls}}$. Plugging this in, the equation for $a_{2ls+4r}$ becomes linear,\dots In the end, we obtain expressions for $a_{2ls+4r},\dots, a_{2ls+2mr}$ which are rational functions of the $c$'s with some powers of $a_{2ls}$ as the denominator. In particular, we have that $a_{2ls+2mr}^2=c_{4ls+4mr}$ which is a polynomial on the Lie algebra. Since $a_{2ls+2mr}$ is a rational function on the Lie algebra and its square is a polynomial, we conclude that $a_{2ls+2mr}$ is a polynomial on the Lie algebra. For $m=1$, we are done. For $m=2$, the remaining constraint equation reads
\[
c_{4ls+4r}=2a_{2ls}a_{2ls+4r} + a_{2ls+2r}^2
\]
Since the other terms are polynomial, we conclude that $a_{2ls+2r}^2$ is polynomial and hence, since $a_{2ls+2r}$ is rational, that $a_{2ls+2r}$ is polynomial. Things become progressively more complicated as we go to higher $m$. For $m=3$, we have that $a_{2ls}$ and $a_{2ls+6r}$ are polynomials on the Lie algebra. For $a_{2ls}\neq 0$, we can solve for the rest of the $a$'s:
\begin{subequations}
\begin{align}
a_{2ls+2r}&= \frac{1}{2a_{2ls}}c_{4ls+2r}\label{m3a}\\
a_{2ls+4r}&=\frac{1}{8a_{2ls}^3}(4c_{4ls+4r}a_{2ls}^2-c_{4ls+2r}^2)\label{m3b}\\
a_{2ls+6r}&=\frac{1}{32a_{2ls}^5}(16c_{4ls+6r}a_{2ls}^4-4c_{4ls+2r}c_{4ls+4r}a_{2ls}^2+c_{4ls+2r}^3)\label{m3c}
\end{align}
\end{subequations}
Since $a_{2ls+6r}$ is a polynomial, \eqref{m3c} can be rewritten as
\[
\begin{split}
32a_{2ls}^5 a_{2ls+6r}&= 16c_{4ls+6r}a_{2ls}^4-4c_{4ls+2r}c_{4ls+4r}a_{2ls}^2+c_{4ls+2r}^3\\
&=16c_{4ls+6r}a_{2ls}^4 -8 a_{2ls}^3 a_{2ls+4r}
\end{split}
\]
or
\[
a_{2ls+4r} = 2 c_{4ls+6r}a_{2ls} -4 a_{2ls}^2 a_{2ls+6r}
\]
which proves that $a_{2ls+4r}$ is a polynomial. Plugging \eqref{m3a} into \eqref{m3b}, we have
\[
a_{2ls+2r}^2 = c_{4ls+4r}-2a_{2ls}a_{2ls+4r}
\]
We just proved that the RHS is a polynomial. Since  $a_{2ls+2r}$ is rational, and its square is a polynomial, it too is a polynomial.

The procedure generalizes to higher $m$, though it becomes increasingly tedious to implement. We first solve for the $a$'s as rational functions of the $c$'s and $a_{2ls}$ with a power of $a_{2ls}$ as the denominator. Noting that $a_{2ls+2mr}$ is a polynomial, we rearrange the expression for it to solve for $a_{2ls+2(m-1)r}$, which we see is a polynomial. Plugging that knowledge into the expression for $a_{2ls+2(m-1)r}$, we learn that $a_{2ls+2(m-2)r}$ is a polynomial, etc. Proceeding in this fashion, we conclude that $a_{2ls+2(m-2)r},\dots, a_{2ls+2r}$ are also polynomials on the Lie algebra.

Away from some codimension-1 locus where $a_{2ls}=0$, the map from the space of $a$'s to the space of $c$'s is invertible. On the locus $a_{2ls}=0$, we have a set of coupled quadratic equations whose number of solutions is finite. If there are more even parts in this block, we continue as before. The $a$'s in the next block are rational functions of the $c$'s and $a_{2ls+2mr}$, with denominators that are powers of $a_{2ls+2mr}$. 

We conclude that, away from a locus of codimension-1, the map $\gamma^{(1)}: S^{(1)}_{\text{CD}}\to S^{(1)}_{\text{BK}}$ is 1-1. On that  locus, it is finite-to-1. Over all of $S^{(1)}_{\text{BK}}$, $\gamma^{(1)}$ is birational and finite -- hence a normalization.

More or less the same story holds for part (ii): a string of even parts at the end of the partition. Consider a partition of the form $[p]=[\dots, (2s)^{2l}]$. The constraint associated to this last even part is
\[
c_{2N-4sl}u^{2l}+c_{2N-2s(2l-1)} u^{2l-1}+\dots+c_{2N-2s}u +\tilde{c}^2 =(a_{N-2sl}u^l+a_{N-2s(l-1)}u^{l-1}+\dots+a_{N-2s}u+a_N)^2
\]
We can take $a_N=\tilde{c}$. Then we find $a_{N-2s}= \tfrac{c_{2N-2s}}{2\tilde{c}}$, and so on. The $a$'s are rational functions of the $c$'s, with denominators which are powers of $\tilde{c}$. The last equation then says that $c_{2N-4sl}= a_{N-2sl}^2$. Since $c_{2N-4sl}$ is a polynomial on the Lie algebra, $a_{N-2sl}$ must be denominator-free (as a rational function on the Lie algebra). Hence it (and the other $a$'s that we solved for) are also polynomials on the Lie algebra. The same considerations as before lead us to conclude that, for this block, $\gamma^{(i)}$ is birational and finite -- hence a normalization.

Finally, when the string of even parts occurs in the middle, the map $\{a_j\}\to S^{(i)}_{\text{BK}}$ is generically 2-1, since if $\{a_j\}$ is a solution, then so is $\{-a_j\}$. A-priori, now the $a_j$ are only algebraic; for instance, the first one $a_{k} = (c_{2k})^{1/2}$. Here, we need to invoke Conjecture \ref{Simpleconjecture} which says that $a_{k}$ is, in fact, a polynomial. We then proceed as before. In the same spirit as the previous cases, the map from the quotient $S^{(i)}_{\text{CD}}/r \to S^{(i)}_{\text{BK}}$ is evidently birational and finite.
\end{proof}

We want to emphasize that the above theorem is purely about the local Hitchin image. We explore some of the subtleties that arise in studying the corresponding questions in the global setting in \ref{globalconsiderations-even}.


\section{Odd Type Constraints}
\label{sec:oddtype}

The constraints of Section \ref{sec:even_and_very_even} had to do with strings of consecutive even parts in $[p]$. We now turn to the constraints associated to consecutive odd parts.




\begin{defn}
(Odd-type constraint) Let $2k=\sum_{i=1}^{2j} p_i$, for some $j$. We say that $c_{2k}$, for $k \lt N$, has an \emph{odd-type constraint} iff $p_{2j}$ and $p_{2j+1}$ are  odd \emph{and} $c_{2k}$ can be written as a perfect square of a polynomial function,  $a_k$, on $\mathfrak{g}$.
\label{odd-type-constraint}
\end{defn}

\begin{defn}(Marked pair) A \textbf{marked pair} is a pair of successive parts, $p_{2j},p_{2j+1}$, where $p_{2j}$ and $p_{2j+1}$ are both odd and $p_{2j} > p_{2j+1}$.
\label{marked-pair-definition}
\end{defn}

It follows from Conjecture \ref{Simpleconjecture} that for every marked pair in the special Hitchin orbit $\mathcal{O}_H$, there is an odd type constraint $c_{2k} = (a_{k})^2$. When such a constraint occurs, we have a choice in the (local) ring of invariant polynomials, we can either continue to use $c_{2k}$ as the generator for the local ring of invariants or replace $c_{2k}$ by $a_k$. 

This situation is somewhat similar to what happens with $c_{2N}$ for any Hitchin orbit with one key difference.  $c_{2N}$ is always a perfect square, whose  square root is the Pfaffian $\tilde{c}_N$. The transformation that flips the sign of $\tilde{c}_N$ is an element of $O(2N)$ but not an element of $SO(2N)$. So, when we are studying the $SO(2N)$ Hitchin system, we have no choice but to choose $\tilde{c}_N$ as the generator for the ring of invariants.

In the case of odd type constraints, we have $c_{2k}=a_k^2$ for some polynomial function $a_k$ and $k \lt N$. The subgroup of $SO(2N)$ which preserves our chosen nilpotent is disconnected. The 
gauge transformations in that subgroup which flip the sign of $a_k$ are not in the identity component. So, one has an \textbf{additional choice} (beyond deciding the group for the Hitchin system) of whether we choose to quotient, locally, by such transformations. Quotienting by such transformations corresponds to choosing $c_{2k}$ as the local invariant while not quotienting corresponds to choosing $a_k$ as the local invariant.

Depending on the choice we make, the corresponding local Hitchin image, which is always a graded vector space, will either have a contribution in degree $2k$ or degree $k$. More generally, when $l$ odd-type constraints occur (each for a different $c_{2k}$) for given Hitchin nilpotent $\mathcal{O}_H$, we have $2^l$ choices for the local Hitchin image. We say that the Hitchin image is spanned by the $a_k$'s when we chose $a_k$s as the local invariants and similarly for the $c_k$s.

\subsection{Sommers-Achar subgroups}


\begin{defn}(the group \textbf{$\overline{A}_b$})
We define a finite group \textbf{$\overline{A}_b(\mathcal{O}_H)$} associated to every special Hitchin orbit $\mathcal{O}_H$. $\overline{A}_b=(\mathbb{Z}_2)^d$, where $d$ is the number of marked pairs in $[\mathcal{O}_H]$. If we label the marked pair $[\dots,(2r+1)_{2j},(2s+1)_{2j+1},\dots]$ by the positive integer $j$ (the subscript denotes the position within the partition), then  the generators of $\overline{A}_b$ are labeled by the ordered set 
\begin{equation}\label{Sdef}
S=\{j_1,j_2,\dots,j_d\}
\end{equation}
with $j_1<j_2<\dots<j_d$.
\label{abarb-definition}
\end{defn} 

$\overline{A}_b(\mathcal{O})$ is a subgroup of Lusztig's quotient, $\overline{A}(\mathcal{O})$ which, in turn, is a subgroup of the equivariant fundamental group $A(\mathcal{O})$ (=component group of the centralizer) of the orbit $\mathcal{O}$. For the reader's convenience, we recall the definitions of these other groups in Appendix \ref{App:Lusztig}. The subscript `b' in $\overline{A}_b$ is meant to signify the fact that it is this subgroup of $\overline{A}$ that is visible in the geometry of the Hitchin base.  If Conjecture \ref{Simpleconjecture} holds, then we can take $B_0$ to be the Hitchin image  parametrized by the $a_{k_i}$. The $i^{\text{th}}$ generator of $\overline{A}_b(\mathcal{O})$ acts on $B_0$ as $a_{k_i}\to -a_{k_i}$. In particular, the Hitchin image parametrized by the $c_{2k}$s is $B_0/\overline{A}_b$.

However, there are intermediate quotients that are of physical interest.

\begin{defn}
A \textit{Sommers-Achar} subgroup $C_i \subset \overline{A}_b(\mathcal{O}_H)$ is a subgroup whose generators are labeled by a subset $S_i \subset S$. Corresponding to the pair $(\mathcal{O}_H,C_i)$, we define the Hitchin image to be $B_0 / C_i$. \label{def_Hitchin_base}
\end{defn}

The above definition is a sharper version of the one given in \cite{Chacaltana:2012zy}. The terminology is motivated by the duality map due to Sommers \cite{MR1850659} and Achar-Sommers \cite{MR1927953} which, in turn, is a refinement of the duality maps of Lusztig-Spaltenstein \cite{spaltenstein2006classes} and Barbasch-Vogan \cite{barbasch1985unipotent}. We should note here that the subgroups $C_i$ arise in \cite{MR1850659},\cite{MR1927953} in a slightly different context. Their relevance to the geometry of the Hitchin image is one of the crucial insights from \cite{Chacaltana:2012zy}. 

\subsection {Special pieces}

Recall that if $[\mathcal{O}]$ is the partition label corresponding to a nilpotent orbit (not necessarily special), then duality acts by $d : \mathcal{O} \rightarrow \mathcal{O}^\vee$ where $[\mathcal{O}^\vee] := D([\mathcal{O}]^T)$. Here, $[\ldots]^T$ denotes the transpose operation on partitions and $D(\ldots)$ is the D-collapse operation. For any $\mathcal{O}$, $\mathcal{O}^\vee$ is always a special nilpotent orbit, so $d$ is neither one--one nor onto. Nevertheless, $d^3=d$.

Here is an example of the duality operation. Take $[\mathcal{O}]=[3^2,2^2,1^2]$. We have $[\mathcal{O}]^T=[6,4,2]$ and $D([\mathcal{O}]^T) = [5^2,1^2]$. We depict the steps using Young diagrams below (where the column sizes encode the parts of a partition) :

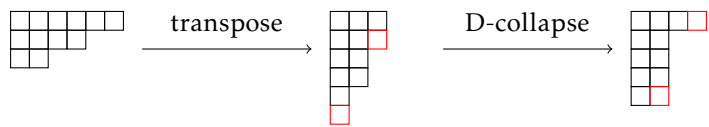
\begin{center}
\begin{tikzpicture}
\tikzset{pics/yb/.style n args={2}{
code= {\draw ({.25*#1},{-.25*#2}) rectangle ++(.25,.25);}}}
\path
pic {yb={0}{0}}
pic {yb={0}{1}}
pic {yb={0}{2}}
pic {yb={1}{0}}
pic {yb={1}{1}}
pic {yb={1}{2}}
pic {yb={2}{0}}
pic {yb={2}{1}}
pic {yb={3}{0}}
pic {yb={3}{1}}
pic {yb={4}{0}}
pic {yb={5}{0}}
(1.75,-.25) edge[->] node[above] {$\footnotesize\text{transpose}$} (4,-.25)
(4.25,0)
pic {yb={0}{0}}
pic {yb={0}{1}}
pic {yb={0}{2}}
pic {yb={0}{3}}
pic {yb={0}{4}}
pic[red] {yb={0}{5}}
pic {yb={1}{0}}
pic {yb={1}{1}}
pic {yb={1}{2}}
pic {yb={1}{3}}
pic {yb={2}{0}}
pic[red] {yb={2}{1}}
(5.75,-.25) edge[->] node[above] {$\footnotesize\text{D-collapse}$} (8,-.25)
(8.25,0)
pic {yb={0}{0}}
pic {yb={0}{1}}
pic {yb={0}{2}}
pic {yb={0}{3}}
pic {yb={0}{4}}
pic {yb={1}{0}}
pic {yb={1}{1}}
pic {yb={1}{2}}
pic {yb={1}{3}}
pic[red] {yb={1}{4}}
pic {yb={2}{0}}
pic[red] {yb={3}{0}}
;
\end{tikzpicture}
\end{center}
\textbf{Note:} We will use $D(..)$ to denote both the operation of D-collapse itself and the resulting special partition. It should be clear from the context which we mean.

For our given special nilpotent, $O_H$, we can ask about its pre-image under $d$.

\begin{defn}(Special piece)
    A special piece consists of those nilpotent orbits in the pre-image of the duality map $d$ which map to the same special nilpotent orbit $O_H$ under $d$. Further, $O_{N_0}\coloneqq d(O_H)$ is the unique special nilpotent in the special piece.
    \label{specialpiece-defn}
\end{defn}

In type $D$, it is known (see Theorem \ref{KPrestated} below) that the size of special pieces are $2^s$ for some natural number $s$. Furthermore,

\begin{defn}
(Non-special degeneration) We define a \emph{non-special degeneration} to be a Kraft-Procesi small degeneration (see Appendix \ref{kpappendix}) that starts with a Nahm orbit $\mathcal{O}$ and results in a non-special orbit $\mathcal{O}'$ which is in the same special piece as $\mathcal{O}$. On the level of the corresponding partitions, a non-special degeneration is
\begin{equation}
\label{eq:non-specialdegen}
\begin{split}
[\dots,(2r+1)_{2j},(2r)_{2j+1},&\dots (2r)_{2j+2m} ,(2r-1)_{2j+2m+1},\dots] \\&\to
     [\dots,(2r)_{2j},(2r)_{2j},\dots,(2r)_{2j+2m},(2r)_{2j+2m+1},\dots].
     \end{split}
\end{equation}
We label it by the positive integer $j$ (again, subscripts denote the location within the partition).
$m=0$ corresponds to a degeneration of type (a); $m>0$ is a degeneration of type (g).
\label{non-special-deformation}
\end{defn}

Starting with the special orbit $O_{N_0}=d(O_H)$, there are $s$ independent non-special degenerations\footnote{see Appendix \ref{kpappendix} for more details on independent non-special degenerations.} that we can do. Let us label them by the set
\begin{equation}\label{Sprimedef}
S'=\{j_1,j_2,\dots,j_s\}
\end{equation}
with $j_1<j_2<\dots<j_s$.
Every orbit in the special piece is obtained by starting with the special orbit $O_{N_0}$ and performing some set of such non-special degenerations. That is to say that the orbits $O_{N_i}$ in the special piece are labeled by subsets $S'_i\subset S'$.

\subsection{Special pieces and Hitchin images}

Associated to each special Hitchin nilpotent, $O_H$, we have defined two ordered sets: the set $S$ of marked pairs and the set $S'$ of non-special degenerations in the dual special piece.

\begin{theorem}
Given a special Hitchin orbit $\mathcal{O}_H$ and the associated ordered sets 
\begin{itemize}
\item $S$ of marked pairs in $[O_H]$ and
\item $S'$ of independent non-special degenerations of $\mathcal{O}_{N_0}=d(\mathcal{O}_H)$,
\end{itemize}
we have:
\begin{enumerate}%
\item[a)] The ordered sets $S$ and $S'$ have the same number of elements. 
\item[b)] Let $\sigma: S' \rightarrow S$ be the order-reversing bijection between them. $\sigma$ induces a bijection between the set of nilpotent orbits $\mathcal{O}_{N_i}$ in the dual special piece and the set of pairs $(O_H,C_i)$. By Conjecture \ref{Simpleconjecture}, we can identify the latter with the set of Hitchin images $B_0/C_i$.
\end{enumerate}
\label{odd-type-theorem}
\end{theorem}

\begin{proof}

Once we have proven part (a), part (b) follows trivially from the definitions we have given. The $\mathcal{O}_{N_i}$ are labeled by subsets $S'_i\subset S'$ and the Hitchin images $B_0/C_i$ are labeled by subsets $S_i\subset S$. So let us prove part (a).

Let $d$ be the number of marked pairs (labelling the generators of $\overline{A}_b(\mathcal{O}_H)$) and let $s$ be the number of independent non-special degenerations of the dual special Nahm orbit $\mathcal{O}_{N_0}=d(\mathcal{O}_H)$. In other words $|S| = d$ and $|S'| = s$. We proceed in two steps:

\begin{itemize}%
\item We show $s \ge d$ starting with a special Hitchin orbit with $\overline{A}_b \neq 1$ and studying its dual special Nahm orbit (Lemma \ref{lemma2} below) .

\item We show $d \ge s$ using Kraft-Procesi's description of the size of the special piece in terms of number of \emph{non-special degenerations} of the dual special Nahm orbit and then deducing the consequences for the special Hitchin orbit (Lemma \ref{lemma3} below).

\end{itemize}
First, we will need some combinatorial facts about how special nilpotent orbits in type-D behave under duality.

\begin{defn}
     We define the notion of a  \textbf{wall} for D-collapse operations. We say that the operation $D(\lambda)=\sigma$ for some partition $\lambda$ has a \textbf{wall} if there is a splitting of $\lambda,\sigma$ of the form $\lambda =\lambda_1 \mid \lambda_2$ , $\sigma =\sigma_1 \mid \sigma_2$ such that each part of $\lambda_2$ is smaller than the smallest part of $\lambda_1$ and
\end{defn}

\begin{displaymath}
D(\lambda_i) = \sigma_i.
\end{displaymath}
Here, the operation $\lambda_1 \mid \lambda_2$ is defined to yield a partition $\lambda$ that is formed by taking together the parts of $\lambda_1$ and $\lambda_2$.

Let $[q]$ be a partition and $[\tilde{q}]=[q]^T$ be its transpose. We will prove that corresponding to every distinct even part in  $[q]$, there is a wall in the operation $D([\tilde{q}])$.

\begin{lemma} (The wall lemma) \\
\label{lemma2}
\begin{enumerate}%
\item Let $\mathcal{O}$ be a special nilpotent orbit. Then, the operation $D([\mathcal{O}]^T)= [\mathcal{O}^\vee]$ had $l$ walls where $l$ is the number of distinct even parts in $[\mathcal{O}]$.
\item The wall corresponding to an even part $r_d^{e_d}$ occurs between the corresponding two parts $\tilde{q}_j,\tilde{q}_{j+1}$ in $[\mathcal{O}]^T$ that obey $\tilde{q}_j - \tilde{q}_{j+1} = e_d$. The wall divides $[\mathcal{O}]^T$ into two pieces $\lambda_1, \lambda_2$ such that $\lambda_1 \mid \lambda_2 = [\mathcal{O}]^T$ and the dual special orbit $d(\mathcal{O})$ is given by $D(\lambda_1) \mid D(\lambda_2)$. When the even part occurs at the begining of $[\mathcal{O}]$ we take $\tilde{q}_{j+1} = 0$ and the wall in $[\mathcal{O}]^T$ is between the last part and $0$.

\end{enumerate}
\end{lemma}

\begin{proof}

Let us also assume that $[\mathcal{O}]$ is not very even. When $[\mathcal{O}]$ is very even, D-collapse acts trivially on $[\mathcal{O}]^T$ and the statement \#1 of Lemma \ref{lemma1} is tautological.

Let us take the case where $[\mathcal{O}]$ has a single distinct even part $t$ occuring somewhere in the middle with multiplicity $e$ for some even number $e$. So, the partition label $[\mathcal{O}]$ is of the form

\begin{displaymath}
[\mathcal{O}]  = [\ldots, t^{e},\ldots ]
\end{displaymath}
and the transpose will be of the form

\begin{displaymath}
[\mathcal{O}]^T = [\tilde{q}_1,\ldots, \tilde{q}_j, \tilde{q}_{j+1},\ldots ]
\end{displaymath}
where $\tilde{q}_j$ and $\tilde{q}_{j+1}$ are even numbers (because $\mathcal{O}$ is special) that obey $\tilde{q}_j - \tilde{q}_{j+1} = e$. Now, since $t$ is the only distinct even part in $[\mathcal{O}]$, it follows that the parts on both sides of $t^e$ in $[\mathcal{O}]$ will be odd parts. This implies that both $\tilde{q}_j$ and $\tilde{q}_{j+1}$ occur with odd multiplicity. Furthermore, the last part of $[\mathcal{O}]$ will be an odd part and this implies that the first part of $[\mathcal{O}]^T$ will be an even part with odd multiplicity. In other words, the transpose partition actually has the form
\begin{displaymath}
[\mathcal{O}]^T = [\tilde{q}_1,\ldots, \tilde{q}_j^{2r+1}, \tilde{q}_{j+1}^{2s+1},\ldots ]
\end{displaymath}
with $\tilde{q}_1,\tilde{q}_j,\tilde{q}_{j+1}$ even. These are the only even parts that occur with odd multiplicity in $[\mathcal{O}]^T$. So, it is clear how D-collapse will act on $[\mathcal{O}]^T$.

\begin{displaymath}
D([\mathcal{O}]^T) = [\tilde{q}_1-1, \ldots, \tilde{q}_j+1, \tilde{q}_j^{2r}, \tilde{q}_{j+1}^{2s},\tilde{q}_{j+1}-1\ldots,1]
\end{displaymath}
So, it follows that there is a \textbf{wall} between $\tilde{q}_j$ and $\tilde{q}_{j+1}$.

The proof for multiple distinct even parts (occuring in the middle) proceeds in a similar way. The crucial point is that in the transpose, the corresponding $\tilde{q}_j$ is always an even part that can not be a source of a D-collapse move. And the part $\tilde{q}_{j+1}$ can not be a recipient.

When an even part $t$ occurs as the last part of $[\mathcal{O}]$, the partition $[\mathcal{O}]^T$ starts with $[\tilde{s}_1^t, \ldots]$ where $\tilde{s}_1$ is an even number equal to the total number of parts in $[\mathcal{O}]$. Under D-collapse, the parts $\tilde{s}_1^t$ will be unchanged. So, it follows that the D-collapse operation $D([\mathcal{O}]^T)$ has a wall after the last occurance of $\tilde{s}_1$ .

When an even part $t$ occurs as the first part of $[\mathcal{O}]$ and with multiplicity $e$, the partition $[\mathcal{O}]^T$ will end with $[\ldots,e]$ and there will be an odd number of even parts with odd multiplicities that are smaller than $e$. So, in the D-collapse operation, the part $e$ will be turned into $e+1$. It follows that there is a wall between $e$ and $0$.
\end{proof}

\paragraph{\textbf{Example}} In terms of Young diagrams, the presence of a \emph{wall} implies that no box is moved across wall during the D-collapse operation. Here is an example with $[\mathcal{O}]=[5^2,4^2,3^2,2^4,1^2]$.

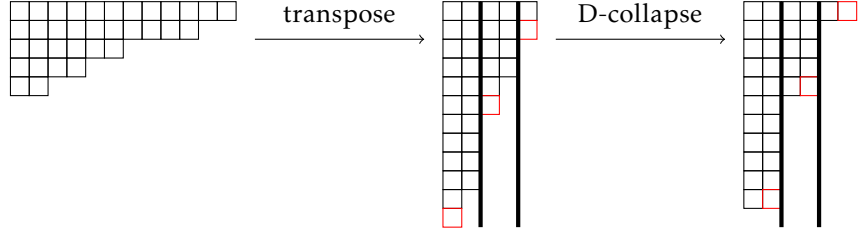
\begin{center}
\begin{tikzpicture}
\tikzset{pics/yb/.style n args={2}{
code= {\draw ({.25*#1},{-.25*#2}) rectangle ++(.25,.25);}}}
\path
pic {yb={0}{0}}
pic {yb={0}{1}}
pic {yb={0}{2}}
pic {yb={0}{3}}
pic {yb={0}{4}}
pic {yb={1}{0}}
pic {yb={1}{1}}
pic {yb={1}{2}}
pic {yb={1}{3}}
pic {yb={1}{4}}
pic {yb={2}{0}}
pic {yb={2}{1}}
pic {yb={2}{2}}
pic {yb={2}{3}}
pic {yb={3}{0}}
pic {yb={3}{1}}
pic {yb={3}{2}}
pic {yb={3}{3}}
pic {yb={4}{0}}
pic {yb={4}{1}}
pic {yb={4}{2}}
pic {yb={5}{0}}
pic {yb={5}{1}}
pic {yb={5}{2}}
pic {yb={6}{0}}
pic {yb={6}{1}}
pic {yb={7}{0}}
pic {yb={7}{1}}
pic {yb={8}{0}}
pic {yb={8}{1}}
pic {yb={9}{0}}
pic {yb={9}{1}}
pic {yb={10}{0}}
pic {yb={11}{0}}
(3.25,-.25) edge[->] node[above] {$\footnotesize\text{transpose}$} (5.5,-.25)
(5.75,0)
pic {yb={0}{0}}
pic {yb={0}{1}}
pic {yb={0}{2}}
pic {yb={0}{3}}
pic {yb={0}{4}}
pic {yb={0}{5}}
pic {yb={0}{6}}
pic {yb={0}{7}}
pic {yb={0}{8}}
pic {yb={0}{9}}
pic {yb={0}{10}}
pic[red] {yb={0}{11}}
pic {yb={1}{0}}
pic {yb={1}{1}}
pic {yb={1}{2}}
pic {yb={1}{3}}
pic {yb={1}{4}}
pic {yb={1}{5}}
pic {yb={1}{6}}
pic {yb={1}{7}}
pic {yb={1}{8}}
pic {yb={1}{9}}
pic {yb={2}{0}}
pic {yb={2}{1}}
pic {yb={2}{2}}
pic {yb={2}{3}}
pic {yb={2}{4}}
pic[red] {yb={2}{5}}
pic {yb={3}{0}}
pic {yb={3}{1}}
pic {yb={3}{2}}
pic {yb={3}{3}}
pic {yb={4}{0}}
pic[red] {yb={4}{1}}
(6.25,.25) edge[ultra thick] (6.25,-2.75)
(6.75,.25) edge[ultra thick] (6.75,-2.75)
(7.25,-.25) edge[->] node[above] {$\footnotesize\text{D-collapse}$} (9.5,-.25)
(9.75,0)
pic {yb={0}{0}}
pic {yb={0}{1}}
pic {yb={0}{2}}
pic {yb={0}{3}}
pic {yb={0}{4}}
pic {yb={0}{5}}
pic {yb={0}{6}}
pic {yb={0}{7}}
pic {yb={0}{8}}
pic {yb={0}{9}}
pic {yb={0}{10}}
pic {yb={1}{0}}
pic {yb={1}{1}}
pic {yb={1}{2}}
pic {yb={1}{3}}
pic {yb={1}{4}}
pic {yb={1}{5}}
pic {yb={1}{6}}
pic {yb={1}{7}}
pic {yb={1}{8}}
pic {yb={1}{9}}
pic[red] {yb={1}{10}}
pic {yb={2}{0}}
pic {yb={2}{1}}
pic {yb={2}{2}}
pic {yb={2}{3}}
pic {yb={2}{4}}
pic {yb={3}{0}}
pic {yb={3}{1}}
pic {yb={3}{2}}
pic {yb={3}{3}}
pic[red] {yb={3}{4}}
pic {yb={4}{0}}
pic[red] {yb={5}{0}}
(10.25,.25) edge[ultra thick] (10.25,-2.75)
(10.75,.25) edge[ultra thick] (10.75,-2.75)
;
\end{tikzpicture}
\end{center}

Corresponding to the two distinct even parts in $[\mathcal{O}]$, we have two wall in the transpose $[\mathcal{O}]^T = [12,10|6,4|2]$. It is clear that no box moves across the walls during the subsequent D-collapse step.

We'll now prove $s \ge d$.

\begin{lemma}
\label{lemma3}
Corresponding to every marked pair in a special nilpotent orbit, there is an independent \emph{non-special degeneration} for the dual special Nahm orbit, i.e. $s\geq d$.
\end{lemma}

\begin{proof}

The statement is quite obvious when the special nilpotent orbit is very even. They have $d=0$ and the dual special piece is always trivial. So, we will not consider these below.

\subsubsection*{\textbf{Very odd orbits}}
We start instead with a \emph{very odd} orbit with $\bar{A}_b = (\mathbb{Z}_2)^d$. This means that it has $d$ marked pairs. Take the first marked pair from the right. Let it be of the form $[\ldots, p_{2j}, p_{2j+1}  ,\ldots]$ where $p_{2j}$ and $p_{2j+1}$ are odd numbers that differ by some even number $2k$ ie $p_{2j}-p_{2j+1}=2k$ for some $k \gt 0$. Now, the $[\mathcal{O}]^T$ will have the following form :

\begin{equation}
\label{veryoddtranspose}
[\mathcal{O}]^T  =  [(2m)^{2n+1},\{\ldots\},(2j)^{2k},\ldots]
\end{equation}
where $2n+1$ is the smallest part of $[\mathcal{O}]$ and $2m$ is the total number of parts in $[\mathcal{O}]$ and the parts in $\{\ldots\}$ are all odd parts with even multiplicities. Now, D-collapse acts in the following way :

\begin{itemize}%
\item If $d=0$, we have $j=0$. And D-collapse will yield $[(2m)^{2n},\ldots,1]$
\item If $d = 1$, D-collapse yields $[(2m)^{2n},\ldots,2j+1,(2j)^{2k-2},2j-1,\ldots,1]$

\end{itemize}
It is clear that the part $2j+1$ in $D([\mathcal{O}]^T)$ will have even parity. By Kraft-Procesi, the dual Nahm special orbit now admits a single small degeneration. If $k=1$, it is a small degeneration of type $(a)$ and if $k \gt 1$, it is a small degeneration of type $(g)$ and the result of the deformation would yield a non-special orbit. We call (Def \ref{eq:non-specialdegen}) such small degenerations as \emph{non-special degenerations}.

Now, when $d \gt 1$, we proceed in a similar way. We handle one marked pair in $[\mathcal{O}]$ at a time. Each of these lead to a single non-special degeneration in $D([\mathcal{O}]^T)$. So, it is clear that $D([\mathcal{O}]^T)$ will have at least $d$ non-special degenerations.

\subsubsection*{\textbf{Orbits with even parts}}
Now, we consider special nilpotent orbits that have even parts. The story is similar except that the even parts in $[\mathcal{O}]$ create walls in the D-collapse of $[\mathcal{O}]^T$.

Within the odd chunks $u_i$ that are between the even parts, the argument proceeds as in the case of very odd orbits.

As a specific case, let us say there is one distinct even part in $[\mathcal{O}]$. Let us also assume that this even part is not the first or the last part. Corresponding to this, there is is a \emph{wall} in $D([\mathcal{O}]^T)$ and there is a splitting of $[\mathcal{O}]^T$ into two pieces $\lambda_1, \lambda_2$ such that $[\mathcal{O}]^T = \lambda_1 \mid \lambda_2$ and each part of $\lambda_2$ is smaller than the smallest part of $\lambda_1$.

The first part of $\lambda_i$ will necessarily be an even part occurring with odd multiplicity. In fact, the parts of $\lambda_i$ will again be as in \eqref{veryoddtranspose}  above. So, the D-collapse of $\lambda_i$ proceeds as in the previous section. Any marked pairs within $\lambda_i$ will contribute the existence of an independent non-special degeneration in $[\mathcal{O}^\vee]$.

So, we get a single non-special degeneration in $[\mathcal{O}^\vee]$ corresponding to every marked pair in $[\mathcal{O}]$.

So, we have shown that Lemma \ref{lemma3} is true, i.e.~$s \ge d$.
\end{proof}

\begin{lemma}
\label{lemma4}
Corresponding to every non-special degeneration in a special Nahm orbit, there is a marked pair in the dual Hitchin orbit, i.e.~$d\geq s$. \end{lemma}

\begin{proof}
Recalling Def \ref{non-special-deformation}, a non-special degeneration takes the form
\begin{equation}
\label{lemma2pf}
[\dots,(2r+1)_{2j},(2r)^{2m},(2r-1)_{2j+2m+1},\dots] \to
     [\dots,(2r)_{2j},(2r)^{2m},(2r)_{2j+2m+1},\dots].
\end{equation}
The transpose of this statement is 
\begin{equation}
\label{lemma2pftranspose}
 [\dots,(2j+2m)_{2r},(2j)_{2r+1},\dots]\to
 [\dots,(2j+2m+1)_{2r},(2j-1)_{2r+1},\dots]
\end{equation}
We now need to distinguish two cases:
\begin{itemize}[leftmargin=75px]
    \item[$m>0$:] a (g)-type non-special degeneration
    \item[$m=0$:] an (a)-type non-special degeneration
\end{itemize}

When $m>0$, by the wall lemma, the D-collapse of the first partition in \eqref{lemma2pftranspose} is
\[
[\dots,(2j+2m+1)_{2r},(2j-1)_{2r+1},\dots],
\]
so in both the source and the target we have the marked pair
\begin{equation}
\label{lemma4themarkedpair}
[\dots,(2j+2m+1)_{2r},(2j-1)_{2r+1},\dots]\quad.
\end{equation}
This, of course, is as it should be: by definition the dual of the two Nahm partitions in \eqref{lemma2pf} is the \emph{same} Hitchin partition.

When $m=0$, we no longer have the wall lemma to show that
\begin{equation}
\label{lemma2pfdog}
D([\dots,(2j)_{2r},(2j)_{2r+1},\dots])=
[\dots,(2j+1)_{2r},(2j-1)_{2r+1},\dots]
\end{equation}
Nevertheless,
\[
D([\dots,(2j+1)_{2r},(2j-1)_{2r+1},\dots])=
[\dots,(2j+1)_{2r},(2j-1)_{2r+1},\dots]
\]
and since the two Nahm partitions in \eqref{lemma2pf} map under duality to the same Hitchin partition, \eqref{lemma2pfdog} continues to hold for the transposes.

Thus each non-special degeneration \eqref{lemma2pf} of  $\mathcal{O}_N$ leads to a marked pair \eqref{lemma4themarkedpair} in $\mathcal{O}_H$.
\end{proof}


Now consider $\mathcal{O}$ a special Hitchin orbit in type-D with $\bar{A}_b(\mathcal{O}) = (\mathbb{Z}_2)^d$. By definition (Def \ref{abarb-definition}) there are precisely $d$ marked pairs.
This proves part (a) of Theorem \ref{odd-type-theorem}. Part (b) follows immediately.
\end{proof}

The 4d $\mathcal{N}=2$ SCFTs of class-S are labeled, in part, by the collection of \emph{Nahm} nilpotent orbits, $\mathcal{O}_N$. The correspondence between Nahm nilpotent orbits $O_{N_i}$ in a special piece and \emph{pairs} $(\mathcal{O}_H,C_i)$ consisting of the dual Hitchin nilpotent $O_H=d(O_{N_i})$ and a ``Sommers-Achar'' subgroup $C_i$ was proposed in  \cite{Chacaltana:2012zy}. Their proposal was based on the observation that the 4d Coulomb branches of the $\mathcal{N}=2$ SCFTs have the form $B_0/C_i$. Theorem \ref{odd-type-theorem} offers a refinement of that proposal, by giving an \textit{a priori} count of the number of orbits in the special piece along with a precise identification of the group $\overline{A}_b(\mathcal{O})$ and its Sommers-Achar subgroups $C_i$.

\subsection{Independence of even and odd type constraints}

We have devoted two separate sections to study the two different phenomena that are new when we study the Hitchin map in type-D. We have presented these are separate discussions without allowing for the possibility that they may interact/interfere with each other. It turn out this never happens.

Say we have a marked pair $(p_{2j},p_{2j+1})$ in the sense of Def \ref{marked-pair-definition}.  Let $2k = \sum_{i=1}^{2j} p_i$. There is an odd-type constraint for the term $c_{2k}$. Since $p_{2j+1}$ is also odd, this guarantees that $c_{2k}$ is not part of any system of even type constraints (of the kind we encountered in Section \ref{sec:even_and_very_even}). So, even when odd and even type constraints occur in the local Hitchin image for the same special orbit $\mathcal{O}$, they necessarily involve mutually distinct subsets of the $c_{2k}$. In this sense, the even and odd type constraints are independent of each other. 

\subsection{Odd type constraints, mass deformations and the Poisson integrable system}\label{sec:poisson}

Throughout this paper, we have confined our study to the geometry of the Coulomb branch underlying a \textit{conformal} $\mathcal{N}=2$ theory arising from a tame Hitchin system. However, there is a class of $\mathcal{N}=2$ preserving relevant deformations called \textit{mass deformations} of the Coulomb branch geometry. Here, we briefly discuss the effect of these deformations on the geometry of the Hitchin system and their close connection with the phenomenon of odd-type constraints. 

A key notion is that of a sheet, which is a connected component of the union of all orbits of a fixed dimension in the Lie algebra. Each such sheet contains a unique nilpotent orbit, sometimes referred to as the boundary of the sheet. However, outside of type A, a given nilpotent orbit can be the boundary of several sheets. This Lie-algebraic observation underlies the multiple Hitchin systems that we have encountered above.

In terms of the Hitchin system, the passage from a conformal theory to a mass deformed theory corresponds to allowing the residue of the Higgs field take generic values in a sheet in the Lie algebra \cite{Balasubramanian:2018pbp} whose boundary is the nilpotent orbit appearing in the description of the conformal limit. When we allow the mass deformation parameters to vary, they act as additional Casimirs (in the sense of \cite{markman1994spectral}) and we get a Poisson integrable system. 

Outside of type A, it is possible for a given nilpotent orbit to occur at the boundary of more than one sheet. Identifying the allowed mass deformations (i.e., identifying the Poisson integrable system in which the Hitchin system of our conformal theory is a symplectic leaf) is tantamount to a choice of one of those sheets. The procedure to identify the exact sheet corresponding to a local mass deformation was outlined in \cite{Balasubramanian:2018pbp}\footnote{Interestingly, it was found that only \textit{special sheets} (a notion defined in \cite{Balasubramanian:2018pbp}) arise as local mass deformations in Class S theories.}. For a given Hitchin nilpotent orbit $\mathcal{O}_H$ we have the following pieces of data which are in 1-1 correspondence
\begin{enumerate}
\item a ``Sommers-Achar'' subgroup $C\subset\overline{A}_b(\mathcal{O}_H)$. For $C=1$, the local Hitchin image has a factor of $\mathbb{C}^s$, on which $\overline{A}_b(\mathcal{O}_H)$ acts; for the other choices of $C$, that factor in the Hitchin image is replaced by $\mathbb{C}^s/C$.
\item a ``Nahm nilpotent orbit'', $\mathcal{O}_N$ in the dual special piece of $\mathcal{O}_H$ (i.e., a nilpotent orbit such that $d(\mathcal{O}_N)=\mathcal{O}_H$)
\item a ``special" sheet whose boundary is $\mathcal{O}_H$.
\end{enumerate}

 The correspondence between (1) and (2) was given in Theorem \ref{odd-type-theorem}, where we related the set of Nahm nilpotent orbits in the dual special piece to ``Achar-Sommers'' subgroups $C\subset \overline{A}_b (\mathcal{O}_H)$. The correspondence between (2) and (3) is that the Bala-Carter Levi of the Nahm nilpotent orbit $\mathcal{O}_N$ is the sheet Levi\footnote{The sheets in a complex Lie algebra $\mathfrak{g}$ are labeled by pairs $(\mathfrak{l},\mathcal{O}_\mathfrak{l})$, where $\mathfrak{l}$ is a Levi subalgebra (which we call the ``sheet Levi'') and $\mathcal{O}_\mathfrak{l}$ is a rigid nilpotent orbit in $\mathfrak{l}$. The (unique) nilpotent orbit $\mathcal{O}$ in $\mathfrak{g}$ which is the boundary of the sheet $(\mathfrak{l},\mathcal{O}_\mathfrak{l})$ is given by orbit induction : $\mathcal{O}=Ind_\mathfrak{l}^\mathfrak{g}(\mathcal{O}_\mathfrak{l}).
 $ (See Appendix A.2 of \cite{Balasubramanian:2018pbp} for a review.)} of the special sheet.

The upshot is that the choice of which odd-type constraints to impose in the conformal case amounts to a choice of what the allowed mass deformations are (i.e., into which Poisson integrable system our conformal theory embeds as a symplectic leaf).

\paragraph{Birational sheets:} In recent work, Losev has developed a theory of \textit{birational sheets} \cite{MR4359565}, which is a refinement of the usual theory of sheets. We expect a close compatibility between Theorem \ref{odd-type-theorem} and Losev's theory of birational sheets. In particular, we expect that the stabilizer of the Sommers-Achar subgroup $C$ (as a subgroup in the component group $A(\mathcal{O})$ ) is related to the degree of the generalized Springer map(s) associated to a particular sheet.  To make things more precise, we formulate the following conjecture.

\begin{conjecture}  Let $\mathcal{O}$ be the special nilpotent orbit occuring in the conformal limit and let $C$ be the Sommers-Achar subgroup associated to one of the Hitchin images corresponding to $\mathcal{O}$ (according to Theorem \ref{odd-type-theorem}). Now, denote by $Z_C$, the stabilizer of $C$ as a subgroup in the component group $A(\mathcal{O})$. Now, let $\mu$ be the generalized Springer map associated to the sheet corresponding to the mass deformed defect (identified as in \cite{Balasubramanian:2018pbp}). Our conjecture is that  
$\text{deg} (\mu) = \vert Z_C \vert $.
\label{degree-conjecture}
\end{conjecture}

The case of birational sheets corresponds to the special case where $Z_C=1$.  Our expectation is that $C=A(\mathcal{O})$ is a necessary and sufficient condition for a mass deformation to correspond to a birational sheet. If we have a non-trivial special piece on the Nahm size,  $C=A(\mathcal{O})$ can possibly occur only for the defect whose Nahm label is the smallest non-special orbit in a special piece. \footnote{It is a fact about the geometry of the special pieces that there is a unique such non-special orbit.} So, our expectation is that, when multiple sheets share a nilpotent orbit as their boundary, the birational sheet corresponds to the one whose sheet Levi is the Bala-Carter Levi of the smallest non-special orbit in the dual special piece.  

\begin{remark}
    Although this paper is confined to discussing the case of type D Hitchin systems, we expect our Conjecture \ref{degree-conjecture} to hold in the $E_6,E_7,E_8$ cases as well. 
\end{remark}

\begin{remark}
    In discussions with E. Sommers and J. Adams, we learnt of earlier work of J. Adams \cite{jeff-adams} where a duality between the smallest non-special orbit and the existence of birational Springer maps was first noticed from the point of view of the theory of unitary representations of non-compact groups. We believe that our conjecture above can offer a different perspective on this duality.
\end{remark}

\section{Global Considerations}
\label{globalconsiderations}

\subsection{Even type constraints}
\label{globalconsiderations-even}
When the even part occurs in middle or at the end, we find a difference between the Hitchin bases in the CD approach and the BK approach. In a global context, these lead to \emph{distinct} Hitchin systems.

There is a third way to associate a Hitchin system to the data of a punctured Riemann surface, where each pucture is ``decorated" by a nilpotent orbit in the Lie algebra, $\mathfrak{g}$: one can construct the class-S theory with that data. This is a 4D $\mathcal{N}=2$ superconformal field theory whose Coulomb-branch geometry is governed by a Hitchin system. That Hitchin system is the one whose Hitchin base is $B_{\text{CD}}$. 

In some cases, we know that $B_{\text{BK}}$ is the base of a \emph{different} Hitchin system which may (or may not) arise from a 4D $\mathcal{N}=2$ theory.

\subsubsection{Example with even part at the end}

Consider the Hitchin system for $\mathfrak{g}=\mathfrak{so}(10)$ on the three-punctured sphere with Hitchin partitions $[3^2,2^2],[3^3,1],[5,3,1^2]$. Before imposing the constraint(s), the naive Hitchin base has dimension
\begin{displaymath}
b_k = (0,0,1,1;0)
\end{displaymath}
for $k=2,4,6,8;5$. The corresponding spectral curve (in the conventions of Appendix \ref{HS3p}) is
\[
\Sigma=\left\{0 = w^{10} +xy(x-y)\bigl(w^4 xy(x-y) c_6+ w^2x^2y^2(x-y) c_8\bigr)\right\}
\]
Since $\tilde{c}=0$, this curve is reducible, with a non-reduced $w^2=0$ and an irreducible octic
\[
\Sigma=\left\{0 = w^2\bigl[w^{8} +x^2y^2(x-y)^2\bigl(w^2 c_6+ xy c_8\bigr)\bigr]\right\}
\]
From the $2^2$ at the end of $[3^2,2^2]$, we have constraints which take the form
\[
\begin{split}
  c_6&=a_3^2\\
  c_8&=2 a_5 a_3\\
  \tilde{c}&= a_5
\end{split}
\]
B\&K would eliminate the $a_k$ and write these as the singular surface in $\mathbb{C}^3$,
\[
c_8^2=4 c_6 \tilde{c}^2
\]
Globally, we have  $\tilde{c}=0$. So this just amounts to saying that $c_8=0$, yielding the spectral curve 
\[
\Sigma_{\text{BK}}=\left\{0 = w^4\bigl[w^{6} +c_6 x^2y^2(x-y)^2\bigr]\right\}
\]
fibered over the complex $c_6$-plane.
C\&D would \emph{solve} the constraints, yielding $c_6=a_3^2,\; c_8=0$. The spectral curve is then
\[
\Sigma_{\text{CD}}=\left\{0=w^4\left(w^3 +i a_3 xy(x-y)\right)\left(w^3 -i a_3 xy(x-y)\right)\right\}
\]
fibered over the complex $a_3$-plane. The sextic curve has become reducible, with the two components being cubics. The map from C\&D's Hitchin base to B\&K's is just the double-cover map $a_3\to a_3^2\equiv c_6$.

The involution $\iota: w\to -w$ exchanges the two cubics. So the Prym of the sextic is the Jacobian of one of the cubics. Indeed, from the work of Argyres et al \cite{Argyres:2016xua} who exhaustively classified all the possible 1-dimensional Hitchin bases, we know that the Hitchin system in this case is an elliptic fibration over the $a_3$-plane with an isolated IV${}^*$ singularity at the origin\footnote{This is the $A_2$ Hitchin system on the 3-punctured sphere with regular nilpotent orbits at each of the punctures.}. 
The entire family of cubic curves has a $\mathbb{Z}_6$ symmetry under which $(x,y,w,a_3)\to(y,x,e^{2\pi i/6}w,a_3)$. Under this $\mathbb{Z}_6$, the Seiberg-Witten differential,
\begin{equation}\label{SWdifferential}
\lambda_{\text{SW}}= \frac{w(xdy-ydx)}{x y(x-y)}
\end{equation}
transforms by the same $6^{\text{th}}$-root of unity. The family has another $\mathbb{Z}_2$ symmetry $(x,y,w,a_3)\to (y,x,w,-a_3)$, under which the Seiberg-Witten differential is invariant\footnote{Since we are describing an \emph{elliptic} fibration, we should specify a choice of section ($x=y$) which is also invariant under the $\mathbb{Z}_2$.}. The quotient by this $\mathbb{Z}_2$ is --- after suitably resolving the singularities of the central fiber --- an elliptic fibration over the $c_6=a_3^2$ plane with an isolated II${}^*$ singularity at the origin. Away from the origin, the fibers are again cubic curves with a $\mathbb{Z}_6$ symmetry under $(x,y,w,c_6)\to(y,x,e^{2\pi i/6}w,c_6)$. The monodromy around the origin has order-3 (instead of order-6), as once around the origin in the $c_6$-plane is twice around the origin in the $a_3$-plane.

The upshot is that CD's prescription yields the IV${}^*$ elliptic fibration, which is the Hitchin system governing the Coulomb branch of the $E_6$ Minahan-Nemeschansy theory, whereas BK's prescription yields the II${}^*$ elliptic fibration, which is the Hitchin system governing the $E_8$ Minahan-Nemeschansky theory. Both of them are perfectly fine 4D $\mathcal{N}=2$ theories; but it is the former that arises in the class-S construction from the puncture data given.

Now consider replacing $[5,3,1^2]$ in the above example with $[5^2]$. Before imposing the constraints, the spectral curve is
\begin{equation}\label{BKglobalEx2spectral}
\Sigma_{\text{BK}}=\left\{0 = w^{10} +x^2y^2(x-y)^2\left(w^4 c_6+ w^2xy c_8+ x^2 y^2 \tilde{c}^2\right)\right\}
\end{equation}
Now, $\tilde{c}$ no longer vanishes and so the B\&K's spectral curve is fibered over the genuinely singular base
\begin{equation}\label{BKglobalEx2base}
B_{\text{BK}}= \{c_8^2=4 c_6 \tilde{c}^2\}\subset \mathbb{C}^3
\end{equation}
C\&D solve the constraints and their Hitchin base is the affine space spanned by $a_3$ and $a_5=\tilde{c}$. Their spectral curve
\[
\Sigma_{\text{CD}} = \Bigl\{ 0 = \bigl(w^5 + i x y(x-y)(a_3 w^2 + \tilde{c} x y)\bigr)\bigl(w^5 - i x y(x-y)(a_3 w^2 + \tilde{c} x y)\bigr) \Bigr\} \to B_{\text{CD}}
\]
Again $\Sigma_{\text{CD}}$ is reducible, with two quintic components, the involution $\iota$ exchanges the two quintics and the Prym of $\Sigma_{\text{CD}}$ is the Jacobian of one of the quintics. This rank-2 Hitchin system is easily recognized to be the one associated to the class-S theory of type   $A_4$ on the 3-punctured sphere with nilpotent orbits $[3,2],[3,2],[5]$ at the punctures (this is called the $R_{2,5}$ SCFT in \cite{Chacaltana:2010ks}\footnote{More precisely, it is the $R_{2,5}$ SCFT with two free hypermultiplets.}).

$B_{\text{CD}}$ is the normalization of $B_{\text{BK}}$. The map $f: B_{\text{CD}}\to B_{\text{BK}}$ is 1-1 everywhere away from the locus $\tilde{c}=0$, where it is 2-1. We know there is a Hitchin fibration over $B_{\text{CD}}$ and this Hitchin system is the one which governs the Coulomb branch geometry of the class-S theory  of type $D_5$, with the above puncture data. What is unclear is whether there is a Hitchin system with base \eqref{BKglobalEx2base} and spectral curve \eqref{BKglobalEx2spectral}
and, if there is, whether it is the Hitchin system associated to some 4D $\mathcal{N}=2$ SCFT.


\subsubsection{Example with even part in the middle}

Consider the Hitchin system for $G=SO(12)$ on a thrice punctured sphere with Hitchin partitions 
$[3^2,2^2,1^2][3^4][5^2,1^2]$.
The spectral curve takes the form
\[\label{BKD6spectralexample}
\Sigma_{\text{BK}}=\left\{0=w^2[w^{10}+x^2y^2(x-y)^2 (c_6 w^4+c_8 xyw^2+c_{10}x^2 y^2)]\right\}
\]
There is a single even-type constraint associated to the $2^2$ in $[3^2,2^2,1^2]$, which yields
\[
c_6 u^2 + c_8 u +c_{10}= (a_3 u + a_5)^2
\]
BK's Hitchin base is the singular quadric
\[
B_{\text{BK}}= \{c_8^2 = 2 c_6 c_{10}\}\subset \mathbb{C}^3
\]
and we have the spectral curve \eqref{BKD6spectralexample} fibered over this base. Of course, this singular quadric is just isomorphic to $\mathbb{C}^2/\mathbb{Z}_2$.

CD's Hitchin base is $B_{\text{CD}}=\mathbb{C}^2$, with coordinates $(a_3,a_5)$. In those coordinates, the spectral curve is (further) reducible
\[
\Sigma_{\text{CD}}= \left\{0=w^2[w^5+ixy(x-y)(a_3 w^2+a_5xy][w^5-ixy(x-y)(a_3 w^2+a_5xy]\right\}
\]
Ignoring the trivial non-reduced $w=0$ component, we see that this is once again the Hitchin system for the $R_{2,5}$ theory. 

As we said, when the even part occurs in the middle of the partition, the map $B_{\text{CD}}\to B_{\text{BK}}$ is generically 2:1. Here, $B_{\text{BK}}$ is the orbifold $B_{\text{BK}}=B_{\text{CD}}/\mathbb{Z}_2$, where $\mathbb{Z}_2: (a_3,a_5)\to (-a_3,-a_5)$

\subsection{Odd type constraints}
\label{globalconsiderations-odd}

When \emph{odd-type} constraints occur, we want to assert that there is a \emph{choice} . We can choose to take $c_{2k}$ as a generator in the (local) ring of invariants or take $a_k$ as the generator. As discussed in Section \ref{sec:poisson}, this corresponds to a choice of what the allowed mass deformations are (or, equivalently, how to embed our symplectic Hitchin system in a Poisson integrable system).

This choice leads to distinct Hitchin systems, as we will illustrate with some examples below. As in the previous subsection, our  conventions are those of  Appendix \ref{HS3p}.

Consider the 3-punctured sphere in the $D_5$ theory, where the nilpotent orbits at $p_1,p_2,p_3$  are given by the Hitchin partitions $[2^4,1^2],[5,3,1^2],[9,1]$. The spectral curve is a homogeneous decic;
Because $\tilde{\phi}=0$, the curve is reducible, with one non-reduced component $w^2=0$, and the other component is a homogeneous octic. Before imposing the constraints, the space of coefficients in the octic is 4-dimensional.

\begin{displaymath}
0 = w^2\left[w^8 + x^2 y(x-y)\left(c_4 w^4 + w^2 x y c_6 + x^2 y(x c_8 - y c'_8)\right)\right]
\end{displaymath}
The constraints associated to $[2^4,1^2]$ impose

\begin{equation}
\begin{aligned}
c_6&= \tfrac{1}{2} c_2 (c_4-\tfrac{1}{4}c_2^2)\\
c'_8&= \tfrac{1}{4} (c_4-\tfrac{1}{4}c_2^2)^2
\end{aligned}
\label{D5exevenconstraints}\end{equation}
But, since $c_2=0$, this yields $c_6=0$ and $c'_8=\tfrac{1}{4}c_4^2$. So the spectral curve is

\begin{equation}
0 = w^2\left[w^8 + x^2 y(x-y)\left(c_4 w^4 + x^2 y(x c_8 -\tfrac{1}{4}y c_4^2)  \right)\right]
\label{D5spectralcurve}\end{equation}
The odd-type constraint associated to $[5,3,1^2]$ is $c_8= a_4^2$.

\begin{itemize}%
\item If we impose that constraint, we obtain the Hitchin system $\mathcal{E}_7\times \mathcal{E}_7$, where $\mathcal{E}_7\to \mathbb{C}$ is the elliptic fibration with an isolated $III^*$ singularity at the origin. Setting\begin{displaymath}
c_4=(b_1+b_2),\qquad a_4=\tfrac{1}{2}(b_1-b_2)
\end{displaymath}
the Hitchin bases for the two copies of $\mathcal{E}_7$ are parameterized by $b_1,b_2$ respectively.

\item If we don't impose the constraint, we obtain the Hitchin system $\operatorname{Hilb}(2,\mathcal{E}_7)$, the Hilbert Scheme of two points on $\mathcal{E}_7$.

\end{itemize}
As another (simpler) example, consider the 3-punctured sphere in the $D_4$ theory with Hitchin partitions $[3^2,1^2],[3^2,1^2],[7,1]$. Here, the spectral curve is

\begin{equation}
0=w^2\left[w^6+x^2 y^2(x-y)(x c_6-y c'_6)
\right]
\label{E6spectral}\end{equation}
The two $[3^2,1^2]$ partitions have odd-type constraints associated to them

\begin{equation}
c_6= (a_1+a_2)^2,\qquad c'_6= (a_1-a_2)^2
\label{E6oddconstraints}\end{equation}
\begin{itemize}%
\item If we impose both constraints, then \cite{Ergun:2020fnm} showed by explicit computation that the Hitchin system we obtain is $\mathcal{E}_6\times \mathcal{E}_6$, where $\mathcal{E}_6\to \mathbb{C}$ is the elliptic fibration with an isolated $IV^*$ singularity at the origin. The Hitchin bases of the two $\mathcal{E}_6$s are parametrized by $a_1$ and $a_2$ respectively.
\item If we impose just one of the constraints, we obtain the Hitchin system $\operatorname{Hilb}(2,\mathcal{E}_6)$.
\item If we impose neither constraint in \eqref{E6oddconstraints}, we obtain a Hitchin system whose succinct description eludes us at present.
\end{itemize}

\section*{Acknowledgements}
We would like to thank David Baraglia, Eric Sommers, Jeffrey Adams, Florian Beck and especially Andres Fernandez Herrero for many helpful discussions. Part of this work was conducted during the Galileo Galilei Institute Workshop on ``Emergent Geometries from Strings and Quantum Fields.'' The work of JD and CP was supported in part by the National Science Foundation under Grant No.~PHY--2210562. CP is also partially supported by the Robert N.~Little Fellowship. The work of AB was done while he was a Simons and Infosys Visiting Research Scholar at ICTS, Bengaluru. The research of RD was partially supported by NSF grant DMS--2001673; by NSF grant DMS--2244978, FRG: Collaborative Research: New birational invariants; and by Simons Foundation Collaboration grant \#390287 ``Homological Mirror Symmetry.''
 
\section*{Declarations}
The authors have no competing interests/conflicting interests to declare that are relevant to the content of this article.

\section*{Appendices}
\appendix

\section{The Equivariant Fundamental Group and Lusztig's Quotient}\label{App:Lusztig}

For the benefit of readers, we recall the formulas that describe the equivariant $G_{ad}$-fundamental group/component group $A(\mathcal{O})$ and Lusztig's quotient $\overline{A}(\mathcal{O})$. 

\begin{itemize}
    \item $A(\mathcal{O})$ in type D (see \cite{CollingwoodMcGovern} ) : Let $a$ be the number of distinct odd parts in the partition label $[\mathcal{O}]$. Then, $A(\mathcal{O}) = (\mathbb{Z}_2)^m$ where $m = \text{max}(a-1,0)$ if all odd parts have even multiplicity and $m = \text{max}(a-2,0)$ otherwise. 
    \item $\overline{A}(\mathcal{O})$ in type D (originally defined in \cite{MR742472}, we use the summary provided in \cite{MR1850659} and \cite{MR1927953})  :  Let $[\mathcal{O}] = [\lambda_1^{a_1},\lambda_2^{a_2},\ldots]$ and let $\nu_i = \sum_{j=1}^{i} a_j $. Now, identify the odd parts in $[\mathcal{O}]$ for which $\nu_i$ is odd. These are sometimes called \textit{marked parts}.  Form a string of integers $\mu$ using the marked parts. Let $m$ be the number of independent, non-empty, even cardinality subsets of $\mu$. We have that $\overline{A}(\mathcal{O}) = (\mathbb{Z}_2)^m$. 
\end{itemize}

Note that the notions of a marked pair in Def \ref{marked-pair-definition} and that of a \textit{marked part} are related but different. Every marked pair has an underlying marked part but not every marked part is part of a marked pair. This will play an important role in what follows. 

From the prescription to compute Lusztig's quotient and the definition of $\overline{A}_b(\mathcal{O})$ in Def \ref{abarb-definition}, it is clear that $\bar{A}_b \subset \bar{A}(\mathcal{O})$. More generally, we have

\begin{displaymath}
\bar{A}_b \subseteq \bar{A}(\mathcal{O}) \subseteq A(\mathcal{O})
\end{displaymath}
Here is a table with some illustrative examples :

\begin{center}
\begin{tabular}{l|l|l|l}
$\mathcal{O}$&$\overline{A}_b$&$\overline{A}$&$A$\\
\hline 
$[3^2,1^2]$&$\mathbb{Z}_2$&$\mathbb{Z}_2$&$\mathbb{Z}_2$\\
$[7,5,3,1]$&$\mathbb{Z}_2$&$\mathbb{Z}_2$&$(\mathbb{Z}_2)^2$\\
$[3^2,2^2,1^2]$&$1$&$\mathbb{Z}_2$&$\mathbb{Z}_2$\\
$[11,9,7,5^2,3^2,1]$&$\mathbb{Z}_2$&$\mathbb{Z}_2$&$(\mathbb{Z}_2)^4$\\
$[11,9,8^2,7,5^2,3^2,1]$&$1$&$\mathbb{Z}_2$&$(\mathbb{Z}_2)^4$\\
$[11,9,7,5,4^2,3^2,1^2]$&$(\mathbb{Z}_2)^2$&$(\mathbb{Z}_2)^4$&$(\mathbb{Z}_2)^5$\\
\end{tabular}
\end{center}

\section{Hitchin Systems for 3-punctured Spheres}\label{HS3p}

Here we set out our conventions for meromorphic Hitchin systems [cite] on a punctured curve $C$. In particular, we will give a very concrete model for type-$D_N$ on a 3-punctured sphere.

Let $D=\sum_i p_i$ be the divisor corresponding to the punctures. A $D_N$ meromorphic Higgs bundle is a pair $(E,\Phi)$, composed of a $D_{N}$-bundle $E$  on $C$ and a holomorphic section $\Phi$ of $ad(E) \otimes L$, where $L=K_C(D)$, called the Higgs field. We fix $\Phi(p_i)\in O_i$, where the $O_i$ are a collection of nilpotent orbits in $\mathfrak{so}(2N)$.

The moduli space of Higgs bundles on $C$,  $\mathcal{M}_{Higgs}$, fibers over the Hitchin base. In type-$A_{N-1}$ \cite{Balasubramanian:2020fwc}, the Hitchin base is the graded vector space $\bigoplus_{k=2}^N H^0\bigl(C,L^{\otimes{k}}(-\sum_i \chi_i^{(k)} p_i)\bigr)$. The $\chi_i^{(k)}$ are positive integers which depend on the choice of nilpotent orbits, $O_i$. For the regular nilpotent orbit $\chi^{(k)}=1$, whereas for the minimal nilpotent orbit  $\chi^{(k)}=k-1$.

More concretely, the sections $\phi_k\in H^0\bigl(C,L^{\otimes{k}}(-\sum_i \chi_i^{(k)} p_i)\bigr)$ are symmetric polynomials in the Higgs field $\Phi$
\[
\phi_k= s_k(\Phi)=\tfrac{1}{k}\operatorname{Tr}(\Phi^k)+\dots
\]
In type-$D_N$, the $\phi_k$ vanish for odd $k$ and $\phi_{2N}=\tilde{\phi}^2$, where
\[
\tilde{\phi}=\operatorname{Pfaff}(\Phi)
\]
As sections of $L^{\otimes{k}}$, the $\phi_k$ vanish to order $\chi^{(k)}_i$ at $p_i$. In type-$D_N$ (in contrast to type-$A$), the leading coefficients of the $\phi_k$ at $p_i$ obey local constraints that are the subject of this paper. The Hitchin base, in type-$D_N$ is the space of solutions to those constraints.

The spectral curve $\Sigma\to C$ is a hypersurface in the total space of $L$
\begin{equation}\label{spectralgen}
\Sigma=\left\{0=w^{2N} +\sum_{k=1}^{N-1} w^{2(N-k)}\phi_k + \tilde{\phi}^2\right\}
\end{equation}
where $w$ is the fiber coordinate on $L$. Every such $D_N$ spectral curve admits an involution $\iota: \Sigma \circlearrowleft$, under which $w$ and the Seiberg-Witten differential are odd
\[
\iota^*(w)=-w,\qquad \iota^*(\lambda_{\text{SW}})=-\lambda_{\text{SW}}
\]
Let $\pi: \Sigma\to \Sigma/\iota$ be the projection; the fibers of the Hitchin map are the Prym variety $\operatorname{Prym}(\pi)$.

For the 3-punctured sphere, we can be even more concrete.
Let $x,y$ be homogeneous coordinates on $C=\mathbb{CP}^1$, such that the three punctures are located at $p_1=\{x=0\}$, $p_2=\{y=0\}$ and $p_3=\{x=y\}$, respectively. The line bundle\footnote{If we take the homogeneous coordinates on $\mathbb{CP}^2$ to be $(x,y,w)$, then the total space of $L$ can be thought of as $\mathbb{CP}^2\backslash(0,0,1)$. } $L=\mathcal{O}(1)\to \mathbb{CP}^1$. The $\phi_k(x,y)$ are then homogeneous polynomials of degree $2k$.
Choosing three arbitrary Hitchin partitions, one for each puncture, the spectral curve takes the form of a homogeneous polynomial of degree $2N$ in the variables $x,y,w$
\begin{equation}\label{3puncturedspectral}
\Sigma= \left\{ 0= w^{2N} +\sum_{k=1}^{N-1}\phi_{2k}(x,y)w^{2(N-k)} +\tilde{\phi}(x,y)^2 \right\}
\end{equation}
Overall, the polynomial in \eqref{3puncturedspectral} is homogeneous, of degree $2N$ in the variables $(x,y,w)$.

For three regular nilpotents, the $\phi_{2k}$ vanish at least linearly at each of the points and we can write
\begin{displaymath}
\phi_{2k}= x y (x-y) f_{2k}(x,y),\qquad \tilde{\phi}= x y (x-y) \tilde{f}(x,y)
\end{displaymath}
where the $f_{2k}$ are homogeneous polynomials of degree $2k-3$ and $\tilde{f}$ is a homogeneous polynomial of degree $N-3$. The Seiberg-Witten differential is a meromorphic 1-form on the total space of $L$. In the case at hand, it has a simple expression in homogeneous coordinates
\begin{displaymath}
\lambda_{SW} = \frac{w(x dy - y dx)}{x y (x-y)}
\end{displaymath}
and is manifestly odd under $\iota$. There is a second $\mathbb{C}^*$ action, under which $w$ scales with weight-1 and the \emph{coefficients} in the polynomials $f_{2k}(x,y)$ scale with weight-$2k$ (and the coefficients in $\tilde{f}(x,y)$ scale with weight-$N$. As a consequence, $\lambda_{SW}$ scales with weight-1. Since the coefficients in the polynomials parameterize the Hitchin base, the latter is graded by this $\mathbb{C}^*$ action and the constraints are homogeneous.





\section{The work of Kraft-Procesi on Special Pieces}
\label{kpappendix}

In this appendix we will will briefly review the main result of  \cite{kraft1989special} and some of the tools used in the proof of this result. 



To begin, we recall that nilpotent orbits in $D_N$ can be designated by $D$-partitions of $2N$, or alternatively by a Young diagram with columns of heights specified by the partition. The only requirement for these partitions is that even integers must appear with even multiplicity, and in cases where all integers are even, there are actually two orbits linked to that partition. Another method of representing these nilpotent orbits involves using partitions of $2N$, where the corresponding Young diagrams have rows of length specified by those partitions (Hitchin). Among these orbits, there is a subset referred to as \textbf{special} which are those in the range of the Spaltenstein map $d$ \footnote{Given by transposing and D-collapsing the Young tableau}. An alternative definition is that a partition $[\mathcal{O}]$ is special if the transpose $[\mathcal{O}]^T$ is a C-partition \cite{CollingwoodMcGovern}\footnote{A partition of $2N$ with the property that odd integers occur with even multiplicity}. An easy way to check this is to check that there are even number of odd parts between even parts or beginning or ending of the partition. For example $[3^2,2^2,1^2]$ is special, while $[3,2^2,1]$ is not.

The Spaltenstein map, $d$ takes Nahm partitions and maps them to Hitchin partitions and vice-versa. However, in type $D$, this map is not a one-to-one correspondence. When two Nahm partitions map to the same Hitchin partition, they are considered to be part of the same \textbf{special piece}. In other words, a special piece is defined as the set of nilpotent orbits that map to the same special nilpotent orbit under the action of $d^2$.

There is a partial-ordering on the nilpotent orbits in $\mathfrak{g}$, given by orbit closure.  A \emph{minimal degeneration} $O\to O'$ is one where $O'$ lies in the closure of $O$ and there are no intervening \emph{between} $O$ and $O'$ in the partial ordering\footnote{More explicitly, there's no third orbit $O''$ such that $O''$ lies in the closure of $O$ and $O'$ lies in the closure of $O''$.}. These minimal degenerations correspond to \emph{moves}, in which we rearrange the boxes in the corresponding Young tableaux.  These moves were classified (a)--(g) in  \cite{kraft1989special}. Moves (b)--(f) involve moving two boxes. Moves (a) and (g) are distinguished by the fact that they involve moving only one box.

\begin{itemize}%
\item \textbf{Type a}: For this move, the $j^{\text{th}}$ column of the partition has odd height $p_j = 2r+1$ and is followed by followed by the $(j+1)^{\text{st}}$ column of height $p_{j+1} = 2r-1$. The move consists of moving one box from the $j^{\text{th}}$ column to the $(j+1)^{\text{st}}$ column. After the move both columns have height $p_j=p_{j+1}=2r$.

\begin{figure}[H]
\begin{center}
\begin{tikzpicture}
\tikzset{pics/yb/.style n args={2}{
code= {\draw ({.50*#1},{-.50*#2}) rectangle ++(.50,.50);}}}
\path
(.2,.50) node[above] {$\scriptsize\text{$j$}$}
(0.8,.50) node[above] {$\scriptsize\text{$j\! +\!1$}$}
(.0,0) 
pic {yb={0}{0}}
pic {yb={1}{0}}
(0.50,- .25) edge[dotted, line width=1pt, dash pattern=on 2pt off 2pt on 2pt]   (.50,-1.50)
 (0,-2.0)
 (0,-2.25)
pic {yb={0}{0}}
pic {yb={0}{1}}
pic[red] {yb={0}{2}}
pic {yb={1}{0}}
(1.50,-1.0) edge[->] node[above] {$\footnotesize\text{a-type}$} (4,-1.0)
(4.45,.50) node[above] {$\scriptsize\text{$j$}$}
(5.05,.50) node[above] {$\scriptsize\text{$j\! +\!1$}$}
(4.25,0)
pic {yb={0}{0}}
pic {yb={1}{0}}
(4.75,- .250) edge[dotted, line width=1pt, dash pattern=on 2pt off 2pt on 2pt]   (4.75,-1.50)
 (4.25,-2.0)
 (4.25,-2.250)
pic {yb={0}{0}}
pic {yb={0}{1}}
pic {yb={1}{0}}
pic[red] {yb={1}{1}}
;
\draw [decorate,
	decoration = {calligraphic brace}] (-.5,-3.25) --  (-.5,0.5);
\path
(-.5,-1.375) node[left] {$2r+1$};
\draw [decorate,
	decoration = {calligraphic brace}] (5.5,0.5) --  (5.5,-2.75);
\path
(6.25,-1.125) node[left] {$2r$};
\end{tikzpicture}
\caption{An example of the (a) move of Kraft-Procesi}
\label{a-move-figure}
\end{center}
\end{figure}
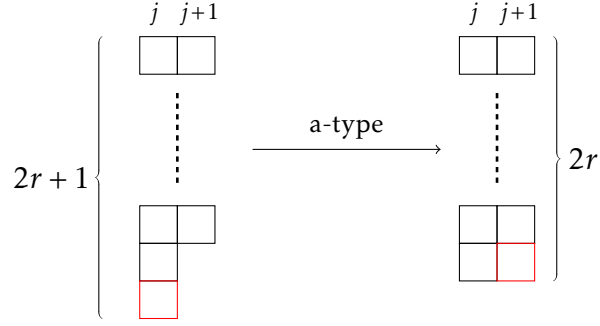

\item \textbf{Type g}: For this move, $p_{j}=2r+1$, $p_{j+1}=p_{j+2}=\dots=p_{j+2m}=2r$ and $p_{j+2m+1}=2r-1$. The move consists of moving one box from the $j^{\text{th}}$ column to the $(j+2m+1)^{\text{st}}$ column. After the move, $p_j=p_{j+1}=\dots=p_{j+2m+1}=2r$.


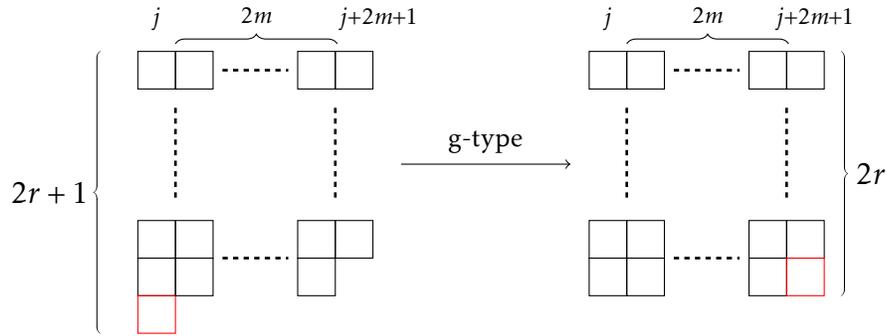
\begin{figure}[H]
\begin{center}
\begin{tikzpicture}
\tikzset{pics/yb/.style n args={2}{
code= {\draw ({.50*#1},{-.50*#2}) rectangle ++(.50,.50);}}}
\path
(.25,.50) node[above=4pt] {$\scriptsize\text{$j$}$}
(3.2,.50) node[above=4pt] {$\scriptsize\text{$j\!\! +\!\! 2m\!\!+\!\!1$}$}
(.0,0) 
pic {yb={0}{0}}
pic {yb={1}{0}}
(0.50,- .25) edge[dotted, line width=1pt, dash pattern=on 2pt off 2pt on 2pt]   (.50,-1.50)
(1.125,0.25) edge[dotted, line width=1pt, dash pattern=on 2pt off 2pt on 2pt]   (2.0,0.25)
(1.125,-2.25) edge[dotted, line width=1pt, dash pattern=on 2pt off 2pt on 2pt]   (2.0,-2.25)
 (0,-2.0)
 (0,-2.25)
pic {yb={0}{0}}
pic {yb={0}{1}}
pic[red] {yb={0}{2}}
pic {yb={1}{0}}
pic {yb={1}{1}}
(2.125,0.0)
pic {yb={0}{0}}
pic {yb={1}{0}}
(2.625,- .25) edge[dotted, line width=1pt, dash pattern=on 2pt off 2pt on 2pt]   (2.625,-1.50)
(2.125,- 2.25)
pic {yb={0}{0}}
pic {yb={0}{1}}
pic {yb={1}{0}}
(3.50,-1.0) edge[->] node[above] {$\footnotesize\text{g-type}$} (5.75,-1.0)
(6.25,.50) node[above=4pt] {$\scriptsize\text{$j$}$}
(9,.50) node[above=4pt] {$\scriptsize\text{$j\!\! +\!\! 2m\!\!+\!\!1$}$}
(6,0) 
pic {yb={0}{0}}
pic {yb={1}{0}}
(6.50,- .25) edge[dotted, line width=1pt, dash pattern=on 2pt off 2pt on 2pt]   (6.50,-1.50)
(7.125,0.25) edge[dotted, line width=1pt, dash pattern=on 2pt off 2pt on 2pt]   (8.0,0.25)
(7.125,-2.25) edge[dotted, line width=1pt, dash pattern=on 2pt off 2pt on 2pt]   (8.0,-2.25)
 (6,-2.0)
 (6,-2.25)
pic {yb={0}{0}}
pic {yb={0}{1}}
pic {yb={1}{0}}
pic {yb={1}{1}}
(8.125,0.0)
pic {yb={0}{0}}
pic {yb={1}{0}}
(8.625,- .25) edge[dotted, line width=1pt, dash pattern=on 2pt off 2pt on 2pt]   (8.625,-1.50)
(8.125,- 2.25)
pic {yb={0}{0}}
pic {yb={0}{1}}
pic {yb={1}{0}}
pic[red] {yb={1}{1}}
;
\draw[decorate, decoration={brace,amplitude=3.7pt}] (0.5,0.6) -- node[above=3.7pt] {$\scriptsize\text{$2m$}$} (2.65,0.6);
\draw[decorate, decoration={brace,amplitude=3.7pt}] (6.5,0.6) -- node[above=3.7pt] {$\scriptsize\text{$2m$}$} (8.65,0.6);
\draw [decorate,
	decoration = {calligraphic brace}] (-.5,-3.25) --  (-.5,0.5);
\path
(-.5,-1.375) node[left] {$2r+1$};
\draw [decorate,
	decoration = {calligraphic brace}] (9.35,0.5) --  (9.35,-2.75);
\path
(10.1,-1.125) node[left] {$2r$};
\end{tikzpicture}
\end{center}
\caption{An example of the (g) move of Kraft-Procesi}
\label{g-move-fig}
\end{figure}

\end{itemize}
These moves are what Kraft-Procesi call \textit{small degenerations}. Two moves are considered \textit{independent} if they can be composed. In other words, if making a move makes it impossible for us to perform a second move, those are considered to be dependent.

To understand what it means for two degenerations to be independent, consider  example the $\mathfrak{so}$(24) partition $\eta = [7^2,5,3,1^2]$. A priori, we can perform three different \textbf{a}-moves, which we'll differentiate by using a subindex:

\begin{displaymath}
\eta \xrightarrow{a_1} \sigma_1 = [7,6^2,3,1^2]
\end{displaymath}
\begin{displaymath}
\eta \xrightarrow{a_2} \sigma_2 = [7^2,4^2,1^2]
\end{displaymath}
\begin{displaymath}
\eta \xrightarrow{a_3} \sigma_3 = [7^2,5,2^2,1]
\end{displaymath}

Looking at $\sigma_2$, we see that we cannot perform another $a_1$-move, hence $a_1$ and $a_2$ are not independent. On the other hand, moves $a_1$ and $a_3$ can be composed in the following way (the order does not matter)

\begin{displaymath}
\eta \xrightarrow{a_1a_3} \sigma_{13} = [7,6^2,2^2,1]
\end{displaymath}
and hence constitute independent degenerations.  

Furthermore, we found it convenient to introduce the term ``non-special degeneration`` (Def \ref{non-special-deformation}) to refer to those small degenerations which lead to non-special orbits. It is only these small degenerations that play an important role in the structure of the special piece.  

Since the terminology around various degenerations can be confusing, we recall them here in one place : 

\begin{itemize}%
\item \emph{Minimal degeneration} (as in Kraft-Procesi \cite{MR694606}) : They are labeled from type (a) to (g).
\item \emph{Small degeneration} (as in \cite{MR694606}) : They are minimal degenerations of type (a) or (g)
\item \emph{Non-special degeneration} (implicit in \cite{MR694606}, we use it explicitly) : These are small degenerations of a special Nahm orbit $\mathcal{O}_s$ that result in a non-special orbit $\mathcal{O}_{ns}$. By Spaltenstein (cited in \cite{MR694606}), we know that $\mathcal{O}_{ns}$ is in the same special piece as $\mathcal{O}_s$.

\end{itemize}

To see the need for the third definition above, consider the special nilpotent orbit $\eta=[7^3,5^2,3^3,1^4]$ in $\mathfrak{so}(44)$ which has three independent small degenerations corresponding to the following operations :

\begin{displaymath}
\eta \xrightarrow{a_1} \sigma_{1} = [7^2,6^2,5^2,3^3,1^4]
\end{displaymath}

\begin{displaymath}
\eta \xrightarrow{a_2} \sigma_{2} = [7^3,5,4^2,3^2,1^4]
\end{displaymath}

\begin{displaymath}
\eta \xrightarrow{a_3} \sigma_{3} = [7^3,5^2,3^2,2^2,1^2]
\end{displaymath}

Out of this, only the third one $\sigma_3$ is a non-special orbit and hence only $a_3$ constitutes a non-special degeneration in the sense of Def \ref{non-special-deformation}.

With this terminology, we can state Theorem 6.2 of Kraft-Procesi \cite{kraft1989special} in the following way:

\begin{theorem}
({Kraft-Procesi}) If a special orbit $\mathcal{O}_s$ admits $s$ \emph{independent} non-special degenerations, then the size of the special piece containing $\mathcal{O}_s$ is $2^s$.
\label{KPrestated}
\end{theorem}


\bibliographystyle{utphys}
\bibliography{refs}
\end{document}